\documentclass[twoside]{article}

\usepackage[colorlinks=true,citecolor=blue,linkcolor=blue,urlcolor=blue]{hyperref}
\usepackage{amssymb}
\usepackage{graphicx}
\usepackage{placeins}
\newcommand{\Real}{\mathbb{R}}
\newcommand{\KL}{\mathrm{KL}}
\usepackage{amsmath,amssymb,amsthm}
\newtheorem{definition}{Definition}
\newtheorem{theorem}{Theorem}
\newtheorem{proposition}{Proposition}

\usepackage[preprint]{aistats2027}
\usepackage[round]{natbib}
\renewcommand{\bibname}{References}
\renewcommand{\bibsection}{\subsubsection*{\bibname}}

\begin{document}
%

%

\runningtitle{Correlated Quantal Response Equilibrium}

\twocolumn[

\aistatstitle{Correlated Quantal Response Equilibrium: \\ A New Equilibrium Notion for Human-AI Strategic Interactions}

\aistatsauthor{
Andy (Zehao) Zhao\textsuperscript{1}
\And
Rahul Jain\textsuperscript{1}
\And
Vijay Subramanian\textsuperscript{2}
}

\aistatsaddress{
\textsuperscript{1}Department of Electrical and Computer Engineering,
University of Southern California\\
\texttt{\{zhaozeha,rahul.jain\}@usc.edu}\\
\textsuperscript{2}EECS Department, University of Michigan, Ann Arbor\\
\texttt{vsgsubram@umich.edu}
}
]

\begin{abstract}
The modeling of the behavior of agents in strategic environments and the prediction of the resulting outcomes has long been a challenging problem; this has been the focus of game theory. The introduction of AI agents as human-aids or as sole decision-makers expands the underlying question: Can we model behavior and predict outcomes of strategic games that involve a mix of AI agents and human players? As an initial answer, we propose Correlated Quantal Response Equilibrium (CQRE) that combines and generalizes classical solution concepts such as Quantal Response Equilibrium and Correlated Equilibrium. We answer questions of existence and propose a method for computing the CQRE that is a best fit to a given joint action distribution. We demonstrate the empirical effectiveness of CQRE predictions using a diverse collection of complex general-sum games, including high-rank and large-action two-player games, challenging multiplayer games, and a large-scale extension of the Selten--Chmura setting to AI--AI and human--AI play. 
\end{abstract}

\section{Introduction}

The early twentieth century roots of game theory lie in modeling human behavior when involved in strategic interactions \citep{von2007theory}, and as a result equilibrium notions \citep{nash1950equilibrium} were introduced as predictions of outcomes. While mathematically elegant, it was then observed that what are now classical solution concepts (or equilibrium) for games fall short of being able to accurately predict outcomes of strategic interactions between human players. This led to many other equilibrium notions (and their refinements) within the game theoretic framework. Eventually, it was understood that human players are bounded rational and not fully rational as assumed in game theory. In response, notions such as Quantal Response Equilibrium (QRE) \citep{mckelvey1995qre} and its variants were developed, which empirically appear to be better predictors of human behavior in strategic interactions. 

The recent use of LLMs as decision aids or even autonomous agents (which we will call `super-rational', somewhat similar to fully rational but still distinct) that take decisions in strategic settings has again highlighted the question of how to model and predict outcomes of AI-AI, human-AI, and even human-human interactions. In this paper, we explore this question, and propose a new solution concept that bring together classical correlated equilibrium with QRE, leading to a new notion of Correlated QRE. Using already available human-human game data \citep{selten2008stationary} and AI-human and AI-AI data that we generate, we show that this notion provides a better empirical fit than extant notions. 


Classical (mixed) Nash equilibrium~\citep{nash1951noncoop} imposes exact best responses and independent randomization between players. Logit-QRE~\citep{mckelvey1995qre,goeree2005regular} softens the best response through payoff-sensitive stochastic choice and is widely used to describe experimental human play~\citep{selten2008stationary}, but it still retains independent randomization between players. Correlated Equilibrium (CE)~\citep{aumann1974subjectivity,aumann1987correlated} allows for coordination and cross-player dependence by studying joint distributions. However, it requires exact obedience to every recommendation and yields a set of distributions instead of a unique prediction. What is missing is a point prediction over joint actions that represents dependence while permitting heterogeneous action-specific deviations from exact obedience.

To that end, we introduce \emph{correlated quantal response equilibrium} (CQRE) and its refinements for general finite $N$-player games. The CQRE extends CE to the setting where the player action choices only need to be close to a best action, instead of one of the best actions as in CE. We achieve this by allowing for slack in how players' evaluate each of their action choices; as this setting allows for per-player and per-action slack, when needed, we also refer to this as CQRE-pa. As with QRE, the slack can be interpreted as a payoff sensitive stochastic choice, or even a tremble. Given a slack vector, CQRE selects the joint distribution with maximum Shannon entropy within a feasible set of distributions, equivalently, the feasible distribution closest to uniform~\citep{jaynes1957information}. This rule yields a unique prediction for every fixed slack vector, which also allows us to show the existence of CQRE. To highlight the flexibility of CQRE-pa, we distinguish two sub-classes of CQRE-pa: CQRE-pp (per player slack) and CQRE-c (common slack) which enforce different types of regularity on the allowed slack.

We evaluate CQRE across human and AI agent play in two- and three-player games. We assess the quality of the prediction using a closeness score---the Kullback-Leibler (KL) distance of the predicted joint action distribution to empirically observed data. Using a  computational scheme that obtains the slack values for the best fit CQRE, we compare predictions with those of other solution concepts. Relative to the best fit Logit-QRE, CQRE proves to be much better on the twelve human plays' games, achieves a better aggregate fit on AI--AI targets, has a better mean fit in six of seven larger-action games, and is favored in all 56 three-player comparisons. This result can be seen in Figure~\ref{fig:sc7-all49-continuous} which clearly shows a better fit to empirical data than logit-QRE on human-human data from twelve games~\citep{selten2008stationary}. However, Logit-QRE remains better in constructed product-form human--AI games. Through our evaluation we also contribute 201,600 accepted one-shot decisions from six OpenAI's GPT-5.x model deployments for the Selten-Chmura games. 



\begin{figure*}[t]
    \centering
    \includegraphics[width=\textwidth]
    {\detokenize{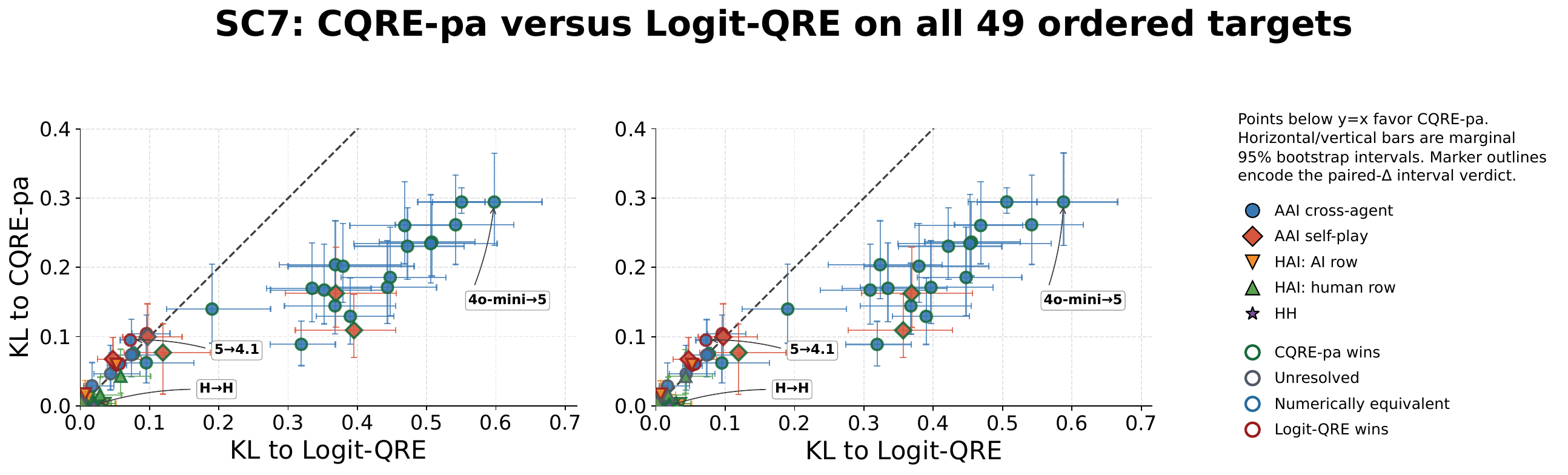}}
    \caption{CQRE-pa versus Logit-QRE on all 49 ordered SC7 targets.
    The Logit-QRE domain is $\beta\in[0,20]$ in the left panel, with
    mean $\Delta=-.0906$ and paired 95\% interval
    $[-.0949,-.0848]$. The right panel uses
    $\beta\in[0,\infty]$, with mean $\Delta=-.0821$ and interval
    $[-.0863,-.0764]$.}
    \label{fig:sc7-all49-continuous}
\end{figure*}

\noindent\textbf{Contributions.} Our contributions are threefold: (1) We introduce a novel solution concept for games: CQRE for general finite-player games that allows for slack in the obedience constraints, and prove its existence and uniqueness for any slack vector. We also establish nesting of common, per-player, and per-player plus per-action slack versions, and derive an exponential-family characterization. (2) We provide a computational solution for the problem of finding the closest CQRE-pa as a bi-level optimization problem with an implementation that is efficient and stable. We use this to show that CQRE-pa provides a much better fit than even logit-QRE for a wide variety of games. (3) The paper's contribution also includes novel LLM decision datasets for various games that could be more generally useful for modeling strategic human-AI and AI-AI interactions behavior; human-human data is from \citep{selten2008stationary}. 


\noindent\textbf{Related work.} CQRE connects to three lines of research. First, QRE and its generalizations model bounded rationality through stochastic or regularized player-level responses~\citep{mckelvey1995qre,goeree2005regular,zhao2026efficient}, but their joint predictions remain products of individual mixed strategies. Second, correlated equilibrium incorporates coordination, where the maximum-entropy correlated equilibrium yields a unique selection from the set of correlated-equilibrium~\citep{aumann1974subjectivity,aumann1987correlated,ortiz2007maximum}. CQRE combines maximum-entropy selection with action-indexed relaxations of obedience. 

A third line of work studies LLMs as strategic players and suggests that LLMs do not exhibit a single stable pattern of strategic behavior: their choices vary with the game, model, prompt framing, and information about opponents, and even strong models sometimes fail to form consistent preferences, update beliefs, or choose payoff-maximizing actions~\citep{akata2023playing,lore2024strategicllms,fan2024gametheory}. These departures from rational choice motivate our use of relaxed obedience conditions, while variation across games and models motivates evaluation across diverse game structures and LLM populations.


Although prior studies evaluate LLM cooperation and coordination in repeated games, pure-coordination tasks, and broader multi-agent game environments~\citep{akata2023playing,agashe2025llmcoordination,xu2024magic}, they do not provide a common equilibrium framework that directly studies the joint action distributions, in particular, dependence between players' action choices, and supports comparisons across human--human, human--AI, and AI--AI play. CQRE is designed to fill this gap.

\section{Correlated Quantal Response Equilibrium}
\label{sec:cqre}
We develop a single probabilistic prediction of joint play by maximum-entropy correlated equilibrium models with relaxed obedience constraints. Starting with a quick review of Logit-QRE and Correlated Equilibrium, we define CQRE by first defining the set of joint distributions that satisfy the relaxed obedience constraints and then select a unique joint distribution by maximizing entropy over this set. We distinguish two slack parameterizations within the general class of slack vectors: a single scalar slack shared by all players and actions which yields CQRE-c; and one slack per player which gives CQRE-pp. Our primary notion is the most flexible per-player and per-action model, which we also call CQRE-pa (when needed); the others (CQRE-c and CQRE-pp) serve as benchmarks. This section characterizes the CQRE distribution for fixed slacks, while Section~\ref{sec:compute} describes how the slack parameters are determined to provide the best fit to joint distributions that arise from observed play.

\noindent\textbf{Notation.} Consider a finite normal-form game $\mathcal{G} =(\mathcal{N},(A_i)_{i\in\mathcal{N}}, (u_i)_{i\in\mathcal{N}})$, where $\mathcal{N}=\{1, \ldots,N\}$ is the set of players, $A_i$ is player $i$'s finite action set, and $u_i:A\to\Real$ is player $i$'s payoff function, with $A:=\prod_{j\in\mathcal{N}}A_j$. An action profile is written as $a=(a_i,a_{-i})$, where $A_{-i}:=\prod_{j\ne i}A_j$. For any finite set $S$, let $\Delta(S)$ denote the probability simplex over $S$. A mixed-strategy profile $\sigma =  (\sigma_i)_{i\in\mathcal{N}}$, with $\sigma_i\in\Delta(A_i)$, induces the product distribution $\bigotimes_{i\in\mathcal{N}}\sigma_i$ over $A$. In contrast, $\rho\in\Delta(A)$ denotes a joint distribution that is potentially correlated. For $a_i,\widetilde a_i\in A_i$ with $\widetilde a_i\ne a_i$, the action profile $(\widetilde a_i,a_{-i})$ represents a unilateral deviation by player $i$ from $a_i$ to $\widetilde a_i$. We denote $\sigma_{-i}:=(\sigma_j)_{j\ne i}$. 

\noindent\textbf{Logit-QRE.} Given opponents' mixed strategies $\sigma_{-i}$, define the expected payoff from action $a_i$ as
\begin{equation}
  v_i(a_i;\sigma_{-i})
  :=
  \sum_{a_{-i}\in A_{-i}}
  \left(\prod_{j\ne i}\sigma_j(a_j)\right)
  u_i(a_i,a_{-i}).
\label{eq:logit-value}
\end{equation}

\begin{definition}[Logit-QRE]
  \label{def:logit-qre}
  Given parameters
  $\boldsymbol{\beta}
  =(\beta_i)_{i\in\mathcal{N}}\in\Real_+^N$, a
  mixed-strategy profile $
  \sigma=(\sigma_i)_{i\in\mathcal{N}}$ is a
  Logit-QRE if, for every $i\in\mathcal{N}$ and $a_i\in
  A_i$,
  \begin{equation}
  \sigma_i(a_i)
  =
  \frac{\exp\{\beta_i v_i(a_i;\sigma_{-i})\}}
  {\sum_{b_i\in A_i}
  \exp\{\beta_i v_i(b_i;\sigma_{-i})\}}.
  \label{eq:logit-qre}
  \end{equation}
  \end{definition}
A Logit QRE's joint prediction satisfies $\rho^{\mathrm{LQRE}} (a)=\prod_{i\in\mathcal{N}}\sigma_i(a_i)$ and is therefore product-form. Consequently, Logit-QRE treats players' randomized actions as independent and cannot capture cross-player dependence represented by a general joint distribution. 

The logit quantal response above is obtained for each $i\in \mathcal{N}$  as the maximizer of 
\begin{equation}
\label{eqn:qr-def}
\max_{\pi_{i} \in \Delta(A_i)} 
 \sum_{a_i\in A_i} \pi_i(a_i) \big[\beta_i v_i(a_i,\sigma_{-i}) - \log(\pi_i(a_i))\big].
\end{equation}
Via a Lagrangian relaxation, we note that \eqref{eqn:qr-def} can also be interpreted as solving for the minimizer of
\begin{align}
\begin{split}
& \min_{\pi_i\in \Delta(A_i)} \sum_{a_i\in A_i}\pi_i(a_i) \log(\pi_i(a_i))\\
& \text{s.t. } \sum_{a_i\in A_i} \pi_i(a_i) v_i(a_i,\sigma_{-i}) \geq \gamma_i, 
\end{split}
\label{eqn:utility_preservation}
\end{align}
for an appropriate $\gamma_i$; there is an equivalence between $\gamma_i$ and $\beta_i$ in \eqref{eqn:qr-def}. We will use the relaxation interpretation in \eqref{eqn:utility_preservation} to define CQRE. 

\noindent\textbf{Correlated Equilibrium.} Unlike Logit-QRE, correlated equilibrium~\citep{aumann1974subjectivity, aumann1987correlated} works with joint distributions $\rho\in\Delta(A)$, so allows coordination between players. For player $i$, recommended action $a_i$, and alternative action $\widetilde a_i\ne a_i$, the recommendation-weighted gain from deviating is defined as
\begin{equation}
  \begin{split}  & R_{i,a_i\rightarrow\widetilde a_i}(\rho)
  :=\\
  & \sum_{a_{-i}\in A_{-i}}
  \rho(a_i,a_{-i})
  \left[
  u_i(\widetilde a_i,a_{-i})
  -u_i(a_i,a_{-i})
  \right].  
  \end{split}
\label{eq:deviation-gain}
\end{equation}
When $\rho_i(a_i):=\sum_{a_{-i}}\rho(a_i,a_{-i})>0$, dividing \eqref{eq:deviation-gain} by $\rho_i(a_i)$ gives the expected gain conditional on receiving recommendation $a_i$.
\begin{definition}[CE]
  \label{def:ce}
  A joint distribution $\rho\in\Delta(A)$ is a
  correlated equilibrium if
  \begin{equation}
  R_{i,a_i\rightarrow\widetilde a_i}(\rho)\le 0, 
  \; 
  \forall\; i\in\mathcal{N},\
  a_i\in A_i,\text{ and }\widetilde a_i\ne a_i.
  \label{eq:exact-obedience}
  \end{equation}
\end{definition}

The obedience holds as follows. A mediator draws an action profile from $\rho$ and privately recommends action $a_i$ to player $i$. Assuming the other players follow their recommendations, obedience requires player $i$ to weakly prefer $a_i$ to every alternative $\widetilde a_i$. When $\rho_i(a_i):=\sum_{a_{-i}}\rho(a_i,a_{-i})>0$, the quantity $R_{i,a_i\rightarrow\widetilde a_i}(\rho)/\rho_i(a_i)$ is the expected payoff gain from deviating conditional on receiving recommendation $a_i$. Correlated  equilibrium requires every such gain be nonpositive. It captures cross-player dependence, but allows no violation of obedience and yields a set of joint distributions instead of a unique prediction.

\noindent\textbf{Correlated QRE.} Observed human and AI behavior does not satisfy every obedience inequality exactly. We therefore replace the zero bound with a nonnegative slack $\gamma_{i,a_i}$ for each player and recommended action with the same slack applied to all possible deviations.
\begin{definition}[CQRE] \label{def:cqre} 
For a nonnegative slack vector $\gamma=(\gamma_{i,a_i})_{i\in\mathcal{N},\,a_i\in A_i} \in\Real_+^{\sum_i|A_i|}$, define
  \begin{equation} K_\gamma := \left\{\rho\in\Delta(A): R_{i,a_i\rightarrow\widetilde a_i}(\rho) \le \gamma_{i,a_i}
\; \forall i,\; a_i, \widetilde a_i\ne a_i \right\}. \label{eq:relaxed-obedience}
\end{equation}
The corresponding correlated quantal response equilibrium is
\begin{equation} \rho_\gamma^\star := \operatorname*{arg\,max}_{\rho\in K_\gamma} \left\{-\sum_{a\in A}\rho(a)\log\rho(a)\right\},  0\log0:=0.
  \label{eq:cqre}
  \end{equation}
\end{definition}

The slack $\gamma_{i,a_i}$ bounds the recommendation-weighted, \emph{ex-ante} deviation gain. The maximum entropy criterion
selects the least concentrated distribution consistent with the relaxed obedience constraints; we note the analogy to the Logit-QRE via \eqref{eqn:utility_preservation}. 

\begin{theorem}[Existence and uniqueness] \label{thm:exists}
For every finite normal-form game and every nonnegative slack vector $\gamma$, problem~\eqref{eq:cqre} has a unique solution.
  \end{theorem}
\begin{proof} The product distribution induced by any mixed Nash equilibrium belongs to the exact correlated-equilibrium set $K_0$, and hence, to $K_\gamma$. Thus, $K_\gamma$ is nonempty. It is also a compact convex polytope as it is the intersection of the probability simplex with finitely many linear inequalities. Shannon entropy is continuous and strictly concave on the simplex, so it attains its maximum at exactly one point in $K_\gamma$.
\end{proof}
Thus every fixed slack vector produces a unique joint prediction, eliminating the selection ambiguity of correlated equilibrium. 

\noindent\textbf{Slack paramterizations.} We consider two restrictions of the slack vector: (i)~\emph{common slack CQRE} (\emph{CQRE-c}), where $\gamma_{i,a_i} =\epsilon$ for every player and action; and (ii)~\emph{per-player slack CQRE} (\emph{CQRE-pp}), where $\gamma_{i,a_i}=\epsilon_i$ for every $a_i\in A_i$. Note that the CQRE definition allows for each $\{\gamma_{i,a_i}\}_{i\in \mathcal{N}, a_i \in A_i}$ to be chosen separately, i.e., in a \emph{per-player} and \emph{per-action} manner.  Hence, when needed, we also use the term \emph{CQRE-pa} for a CQRE. The scalar and per-player specifications will serve as restricted benchmarks.

\begin{proposition}[Nesting] \label{prop:nesting}
Every CQRE distribution obtainable under the scalar parameterization is also obtainable under the per-player parameterization, and every distribution obtainable under the per-player parameterization is also obtainable under the per-player and per-action parameterization.
\end{proposition}
\begin{proof} A scalar slack is embedded in the per-player specification by setting $\epsilon_i=\epsilon$ for every player. A per-player slack is embedded in the per-action specification by setting $\gamma_{i,a_i}=\epsilon_i$ for every $a_i\in A_i$. Each embedding leaves the slack vector, feasible set, and unique maximum-entropy distribution unchanged.
\end{proof}
When $\gamma=0$, the relaxed obedience set $K_0$ coincides with the correlated-equilibrium set, and $\rho_0^\star$ is its unique maximum-entropy element. To characterize the opposite limiting end point, define the largest possible gain from deviating from action $a_i$ as
\begin{equation} \overline d_{i,a_i} := \max\left\{0,\
\max_{\substack{\widetilde a_i\ne a_i\\a_{-i}\in A_{-i}}}
\left[ u_i(\widetilde a_i,a_{-i})-u_i(a_i,a_{-i}) \right] \right\}.
\label{eq:max-deviation}
\end{equation}
If $\gamma_{i,a_i}\ge\overline d_{i,a_i}$ for every player-action pair, all obedience constraints are redundant. Consequently, $K_\gamma=\Delta(A)$ and $\rho_\gamma^\star$ is the uniform distribution over joint action profiles.
More generally, if $\gamma\le\gamma'$ component-wise, then $K_\gamma\subseteq K_{\gamma'}$. Hence the maximum attainable entropy is non-decreasing as the obedience constraints are relaxed. 

\begin{proposition}[Exponential-family form] \label{prop:exp}
Suppose $\rho_\gamma^\star$ has full support. Then  

\begin{equation}
\begin{split}
  & \rho_\gamma^\star(a) = \\
  & \frac{1}{Z(\lambda)} 
  \exp\left\{ 
  \sum_{i\in\mathcal{N}} \sum_{\widetilde a_i\ne a_i} \lambda_{i,a_i,\widetilde a_i}
  \left[u_i(a_i,a_{-i}) -u_i(\widetilde a_i,a_{-i}) \right] 
  \right\},
  \end{split}
  \label{eq:gibbs}
\end{equation}
where $Z(\lambda)$ is the normalizing constant and $\lambda_{i,a_i,\widetilde a_i}\ge0$ is the KKT multiplier associated with the corresponding obedience constraint. These multipliers satisfy
\begin{equation}
  \lambda_{i,a_i,\widetilde a_i} \left[ R_{i,a_i\rightarrow\widetilde a_i}(\rho_\gamma^\star) -\gamma_{i,a_i} \right] =0.
  \label{eq:complementary-slackness}
\end{equation}
\end{proposition}

\begin{proof} Apply the KKT conditions to the entropy-maximization problem over $K_\gamma$. Stationarity with respect to each positive probability $\rho(a)$ gives the exponential form, while dual feasibility and complementary slackness give the stated conditions.
\end{proof}
A complete proof, including the boundary-support version of the KKT characterization, is provided in Appendix~\ref{app:theory}. The full-support assumption is needed because a finite exponential representation assigns positive probability to every profile. When $\rho_\gamma^\star$ lies on the boundary of the simplex, the corresponding representation applies on its support.

\noindent\textbf{Comparison to Logit-QRE beyond the analogy.} Equations~\eqref{eq:logit-qre} and~\eqref{eq:gibbs} both have exponential forms, but their parameters have different meanings. In Logit-QRE, $\beta_i$ is a model parameter that modulates the expected payoffs of the actions available to player $i$; larger values make the player's response more concentrated on higher-payoff actions. In CQRE, $\gamma_{i,a_i}$ is a model parameter that sets the maximum permitted recommendation-weighted deviation gain from action $a_i$. Once $\gamma$ is fixed, the inner entropy maximization determines the multipliers $\lambda_{i,a_i,\widetilde a_i}$. These multipliers are not fitted behavioral parameters, instead they indicate which (relaxed) obedience constraints influence the maximum-entropy solution. An inactive constraint has the multiplier being zero, whereas an active constraint may have a positive multiplier. Therefore, $\beta$, $\gamma$, and $\lambda$ do not have a direct analogy. Moreover, Logit-QRE always produces a product distribution, and CQRE being defined directly over joint action profiles, extends beyond product distributions.


    

\section{Determining best-fit CQRE}
\label{sec:compute}
Section~\ref{sec:cqre} introduced common, per-player and per-action parameterizations of relaxed obedience. We focus here on CQRE-pa, while retaining the other two as benchmarks. Relaxing obedience introduces a vector of unknown slack parameters: for every fixed $\gamma$, the CQRE inner problem selects the unique maximum-entropy distribution $\rho_\gamma^\star$, but theory does not determine which $\gamma$ best describes observed behavior. Estimating the slack from observed play leads to a bi-level optimization problem.

 Given a (empirical) target distribution $\hat p\in\Delta(A)$, the inner problem computes $\rho_\gamma^\star$ for each candidate slack vector, and the outer problem selects
\begin{equation}
  \widehat{\gamma} \in \operatorname*{arg\,min}_{\gamma\in\Gamma_{\mathrm{pa}}} \KL\!\left(\hat p\,\middle\|\, \rho_\gamma^\star\right),
  \label{eq:outer-problem}
\end{equation}
where $\Gamma_{\mathrm{pa}}$ is the compact per-action slack domain to be discussed in the following. The resulting CQRE-pa prediction is $\widehat{\rho}:=\rho_{\widehat{\gamma}}^\star$. This is not a joint optimization over $(\rho,\gamma)$: for every candidate $\gamma$, the inner maximum-entropy problem must first determine $\rho_\gamma^\star$.

Because each obedience slack bounds a recommendation-weighted expected payoff gain, $\gamma_{i,a_i}$ is measured in the same utility units as the game's payoff differences. We rescale each game before fitting CQRE. Let
\begin{equation}
  s_{\mathcal G} := \max_{\substack{i\in\mathcal N,\,
  a_i,\widetilde a_i\in A_i\\
  a_{-i}\in A_{-i}}}
  \left| u_i(a_i,a_{-i}) -u_i(\widetilde a_i,a_{-i}) \right|.
\end{equation}
When $s_{\mathcal G}>0$, we divide every payoff by $s_{\mathcal G}$; when $s_{\mathcal G}=0$, no rescaling is needed. Consequently, every unilateral payoff difference has magnitude at most one. Adding a constant to a player's payoffs does not affect any obedience constraint. In the remainder of this section, $u_i$ denotes the rescaled payoff function.

\noindent\textbf{Matrix representation of relaxed obedience.} Let $\mathcal{D}:=\{(i,b_i,\widetilde b_i): i\in\mathcal{N},\ b_i,\widetilde b_i\in A_i,\ \widetilde b_i\ne b_i\}$
index the obedience constraints, and let $\mathcal{I}:=\{(i,b_i):i\in\mathcal{N},\ b_i\in A_i\}$ index the per-action slack coordinates. Their dimensions are $d_{\mathcal D}=\sum_{i\in\mathcal N}|A_i|(|A_i|-1)$ and $d_{\mathrm{pa}}=\sum_{i\in\mathcal N}|A_i|$, respectively. After fixing an ordering of the joint profiles, write $\rho=(\rho(x))_{x\in A}\in\Real^{|A|}$. Define the obedience matrix $C\in\Real^{d_{\mathcal D}\times |A|}$ by
\begin{equation}
  C_{(i,b_i,\widetilde b_i),x} := \mathbf{1}\{x_i=b_i\}
  \left[ u_i(\widetilde b_i,x_{-i})-u_i(b_i,x_{-i}) \right].
  \label{eq:obedience-matrix}
\end{equation}
Then
\begin{equation}
  (C\rho)_{(i,b_i,\widetilde b_i)} = R_{i,b_i\rightarrow\widetilde b_i}(\rho).
\end{equation}

Let $D_{\mathrm{pa}}\in\{0,1\}^{d_{\mathcal D}\times d_{\mathrm{pa}}}$ replicate each player-action slack across all deviations from that action:
\begin{equation}
(D_{\mathrm{pa}})_{(i,b_i,\widetilde b_i),(j,c_j)} := \mathbf{1}\{i=j,\ b_i=c_j\}.
\label{eq:slack-replication}
\end{equation}
The relaxed obedience set can therefore be written as
\begin{equation}
  K_\gamma = \left\{ \rho\in\Real_+^{|A|}:
  \mathbf{1}_{|A|}^{\mathsf T}\rho=1,\quad
  C\rho\le D_{\mathrm{pa}}\gamma \right\}.
  \label{eq:matrix-feasible-set}
\end{equation}

Here $\mathbf{1}_{|A|}$ is the $|A|$-dimensional all-ones vector, and the inequality is interpreted componentwise.

Equation~\eqref{eq:max-deviation} shows that any value $\gamma_{i,a_i}>\overline d_{i,a_i}$ is redundant: replacing it by $\overline d_{i,a_i}$ leaves the feasible set unchanged. We may therefore restrict the per-action slack vector to the compact
domain
\begin{equation}
  \Gamma_{\mathrm{pa}}
  :=\prod_{(i,b_i)\in\mathcal I}[0,\overline d_{i,b_i}].
  \label{eq:pa-domain}
\end{equation}

For any nonnegative slack vector $\gamma$, define its componentwise clipped version by $\widetilde\gamma_{i,a_i} :=\min\{\gamma_{i,a_i},\overline d_{i,a_i}\}$. Then $\widetilde\gamma\in\Gamma_{\mathrm{pa}}$, $K_{\widetilde\gamma}=K_\gamma$, and therefore $\rho_{\widetilde\gamma}^\star=\rho_\gamma^\star$. Hence restricting \eqref{eq:outer-problem} to $\Gamma_{\mathrm{pa}}$ does not change the set of CQRE-pa predictions.

\noindent\textbf{Inner maximum-entropy problem.} For a fixed slack vector $\gamma\in\Gamma_{\mathrm{pa}}$, the CQRE inner problem can be written in matrix form as
\begin{equation}
  \begin{aligned}
  \rho_\gamma^\star \in\operatorname*{arg\,min}_{\rho\in\Real_+^{|A|}:  \mathbf{1}_{|A|}^{\mathsf T}\rho=1} \quad
  & \sum_{x\in A}\rho(x)\log\rho(x)\\
  \text{subject to}\quad&
  C\rho\le D_{\mathrm{pa}}\gamma,
  \end{aligned}
  \label{eq:inner-matrix-program}
\end{equation}
with the convention $0\log0=0$. The program has $|A|=\prod_{i\in\mathcal N}|A_i|$ probability variables and $d_{\mathcal D}=\sum_{i\in\mathcal N}|A_i|(|A_i|-1)$ obedience inequalities. Its objective is convex and its constraints are linear; Theorem~\ref{thm:exists} therefore guarantees a unique optimizer.

Our implementation solves \eqref{eq:inner-matrix-program} using sequential least-squares programming~\citep{kraft2005slsqp}. The solver returns the unique distribution $\rho_\gamma^\star$ used to evaluate the outer loss in \eqref{eq:outer-problem}. Numerical tolerances, initialization, and feasibility checks are reported in Appendix~\ref{app:cqre-computation}.

\noindent\textbf{Outer loss.} For an empirical target $\hat p\in\Delta(A)$, define
\begin{equation}
  L_{\mathrm{pa}}(\gamma) :=
  \KL\!\left(\hat p\,\middle\|\,\rho_\gamma^\star\right)=
  \sum_{\substack{x\in A\\\hat p(x)>0}}
  \hat p(x)\log
  \frac{\hat p(x)}{\rho_\gamma^\star(x)}.
  \label{eq:outer-loss}
\end{equation}
The fitted slack vector $\widehat\gamma$ minimizes $L_{\mathrm{pa}}$ over $\Gamma_{\mathrm{pa}}$, as stated in \eqref{eq:outer-problem}. Appendix~\ref{app:outer-theory} establishes existence of an outer minimizer and gives the required continuity and regularity results.

\noindent\textbf{Outer optimization.} At points where the active obedience constraints remain unchanged and the inner KKT system satisfies the required regularity conditions, we differentiate that system implicitly to obtain $J_\gamma:=\partial\rho_\gamma^\star/\partial\gamma$. The gradient outer objective is
\begin{equation}
  \nabla_\gamma L_{\mathrm{pa}}(\gamma) =
  -J_\gamma^{\mathsf T}
  \left( \hat p\oslash\rho_\gamma^\star \right),
  \label{eq:outer-gradient}
\end{equation}
where $\oslash$ denotes component-wise division. Appendix~\ref{app:implicit-gradient} derives the KKT-based expression for $J_\gamma$. We use the analytic gradient only when the feasibility, complementarity, rank, conditioning, and full-support checks pass. Otherwise, the algorithm uses bounded finite differences, with one-sided differences at the boundary of $\Gamma_{\mathrm{pa}}$. Changing the slack vector can alter which inner obedience constraints bind, so the outer loss is not guaranteed to be convex. We first search over a fixed grid of candidate slack values. Starting from the best grid solution and additional feasible points, we then run multistart limited-memory Broyden--Fletcher--Goldfarb--Shanno optimization with box constraints (L-BFGS-B)~\citep{byrd1995limited}. The grid solution is retained as an incumbent, so continuous optimization cannot return a worse fit up to numerical tolerance. Complete details of the discrete initialization and continuous optimization are provided in Appendices~\ref{app:outer-domain-computation} and~\ref{app:continuous-outer-computation}.

\noindent\textbf{Other equilibrium benchmarks.} Besides CQRE, we evaluate five equilibrium benchmarks. Common and per-player CQRE are fitted using the same continuous bi-level procedure described earlier and serve as ablations of per player and per-action specific slack. By Proposition~\ref{prop:nesting}, their best attainable in-sample KL cannot be lower than that of CQRE-pa. The CE selection we use is the maximum-entropy CE obtained by setting $\gamma=0$. The mixed-Nash equilibrium benchmark is a product mixed-strategy equilibrium; when multiple mixed Nash equilibria exist, we select the one with the lowest KL divergence to the target.

For Logit-QRE, let $\rho_\beta^{\mathrm{LQRE}}$ denote the product-form solution of Definition~\ref{def:logit-qre} at $\beta$, computed using the original game payoffs. For each target, we fit
\begin{equation}
  \widehat\beta
  \in
  \operatorname*{arg\,min}_{\beta\in[0,\infty]}
  \KL\!\left(
  \hat p\,\middle\|\,
  \rho_\beta^{\mathrm{LQRE}}
  \right).
\label{eq:logit-fitting}
\end{equation}
The case of $\beta=\infty$ denotes the limiting mixed-Nash prediction on the selected solution branch. As a robustness check, we also fit Logit-QRE over $\beta\in[0,20]$. Appendix~\ref{app:logit-computation} specifies the parameter domains, branch-selection rule, equilibrium enumeration, and numerical procedures for all benchmarks.


\section{Empirical Results}
\label{sec:empirical}

\paragraph{Datasets.}\label{par:datasets}
We evaluate CQRE-pa on four benchmark panels spanning human and LLM behavior, two- and three-player games, and varying action-space and payoff structure. The first panel uses the published Selten--Chmura human data. The remaining panels contain LLM decisions collected for this study: seven strongly-monotone and rank-$r$ two-player games, a three-player extension, and an all-12 Selten--Chmura extension. Our principal dataset contribution is the all-12 Selten--Chmura extension, containing 201,600 valid one-shot decisions from six OpenAI model deployments. Combined with published human frequencies, these data support 588 separately fitted target distributions: 432 AI--AI, 144 constructed directional human--AI, and 12 human--human targets.

We define a \emph{target} as an empirical joint-action distribution for a fixed game, agent-to-role assignment, and experimental condition. Appendix~\ref{app:data} documents the game constructions, prompt structure, model deployments, collection settings, target-construction rules, and acceptance and retry procedures.

\noindent\textbf{Evaluation setup.} Using the panels described in \hyperref[par:datasets]{Datasets}, we investigate four questions: (i) whether CQRE-pa improved descriptive fit to established human behavior; (ii) how its relative fit changes across human--human, AI--AI, and constructed human--AI targets for different games; (iii) whether its performance extends to games with larger action spaces and more than two players; and (iv) which populations and game structures favor CQRE-pa instead of other notions, such as, Logit-QRE.

\paragraph{Evaluation method.} 
\label{par:evaluation-method}
For each empirical target distribution $\hat p$, we fit CQRE-pa and the equilibrium benchmarks defined in Section~\ref{sec:compute} to the same data. Our primary comparison uses
\begin{equation}
  \Delta := \KL\!\left(\hat p\middle\|
  \widehat\rho_{\mathrm{CQRE\text{-}pa}\!}\right) -
  \KL\!\left(\hat p\middle\|
  \widehat\rho_{\mathrm{Logit\text{-}QRE}\!}\right),
  \label{eq:empirical-delta}
\end{equation}
so $\Delta<0$ favors CQRE-pa. We report synchronized paired 95\% bootstrap intervals, refitting both models within each replicate. Because the observational units differ across panels, each subsection specifies its resampling unit and number of replicates. These results measure descriptive in-sample fit. 


\subsection{Established Human Behavior} \label{sec:selten}
We first evaluate CQRE-pa on the twelve Selten--Chmura human games. For each game, the target is the equal-weight average of its observed within-group joint-action distributions. To preserve the dependence structure of the original experiment, each of $B=2000$ hierarchical bootstrap replicates resamples entire subject groups and then resamples the observed action pairs within each selected group. Figure~\ref{fig:selten-human-kl} visualizes the CQRE-pa and Logit-QRE fits, while Appendix Table~\ref{tab:selten} reports the complete per-game results and paired confidence intervals.

 CQRE-pa provides the best aggregate fit to established human behavior. Its mean KL is $0.0040$, compared with $0.0093$ for Logit-QRE, corresponding to a 56.9\% reduction. Fit also improves across the nested CQRE specifications, from $0.0072$ for CQRE-c to $0.0057$ for CQRE-pp and $0.0040$ for CQRE-pa, demonstrating the value of action-specific slack. CQRE-pa has the lowest point KL in nine of the twelve games; At the aggregate level, the paired interval for $\Delta$ is strictly negative, establishing that CQRE-pa fits the twelve-game human benchmark than Logit-QRE.

\subsection{AI vs AI and Human vs AI extension}
  \label{sec:selten-all12}
We next use the all-12 Selten--Chmura extension to compare agent populations while holding the underlying games fixed. We cross six AI models with one another and with the human benchmark for each game, producing $7\times7=49$ ordered row-agent--column-agent targets. Across twelve games, the resulting 588 targets comprise 432 AI--AI (AAI), 144 human--AI (HAI), and 12 human--human (HH) targets. AAI targets use paired AI outcomes, HH targets preserve the observed human joint distributions, and HAI targets are constructed by corresponding human and AI marginals.

\begin{table}[t]
  \centering
  \scriptsize
  \setlength{\tabcolsep}{2.5pt}
  \caption{All-12 Selten--Chmura results by agents. Model
  columns report mean KL. Negative $\Delta$
  favors CQRE-pa;
  intervals are paired 95\% bootstrap.}
  \label{tab:selten-populations}
  \begin{tabular}{@{}lrrrl@{}}
  \hline
  Population & Targets & CQRE-pa & Logit-QRE & $\Delta$
  [95\% CI] \\
  \hline
  HH  & 12  & .00400 & .00929 & $-.0053\;
  [-.0092,-.0008]$ \\
  AAI & 432 & .23526 & .24253 & $-.0073\;
  [-.0082,-.0063]$ \\
  HAI & 144 & .10011 & .08896 & $+.0112\;[+.0098,
  +.0124]$ \\
  \hline
  All & 588 & .19744 & .20016 & $-.00272\;
  [-.00354,-.00192]$ \\
  \hline
  \end{tabular}
  \end{table}

The relative advantage changes systematically across agent populations. CQRE-pa reduces mean KL by 56.9\% on HH targets and by 3.0\% on AAI targets, while its target-weighted mean KL across all 588 targets is 1.36\% lower than that of Logit-QRE. Constructed HAI targets exhibit the opposite pattern, with Logit-QRE attaining 11.1\% lower mean KL. This contrast is consistent with CQRE-pa's joint and action-specific flexibility being useful for HH and AAI outcomes, whereas the product-form Logit-QRE is naturally aligned with HAI targets constructed as products of human and AI marginals. We examine SC7 as a concrete example of this pattern. All four population and orientation aggregates have negative mean $\Delta$: $-0.1092$ for AAI, $-0.0069$ for human-row HAI, $-0.0052$ for AI-row HAI, and $-0.0188$ for HH. Complete game-, population-, and model-level results are reported in Appendix~\ref{app:all12-results}.

\subsection{Larger-Action Games}
  \label{sec:multigame}

\noindent\textbf{Games.} We next test whether CQRE-pa's performance extends to larger action spaces and different payoff structures. Adapting the structured game families studied by \citet{zhao2026efficient}, we consider (i) strongly-monotone $n{\times}n$ games with $n\in\{3,5,7\}$ and (ii) general-sum $5{\times}5$ games generated from rank-$r$ payoff components, with $r\in\{1,2,3,4\}$. For each game, six AI models form 21 pairings, with 200 independent plays per pairing.

\noindent\textbf{Strongly-monotone games.}
For each $n\in\{3,5,7\}$, we draw a Gaussian matrix, extract and normalize its antisymmetric component $K$, and set $R_1=0.2I_n+\sqrt{n}K$ and $R_2=0.2I_n-\sqrt{n}K$. Thus $R_1+R_2=0.4I_n$ is positive definite, while $K$ introduces cyclic interactions between actions.

\noindent\textbf{Rank-$r$ general-sum games.} For each $r\in\{1,2,3,4\}$, we draw $U,V\in\Real^{5\times r}$ and form the rank-at-most-$r$ component $M=UV^{\mathsf T}$. Combining $M$ with an independently generated normalized antisymmetric matrix $K$, we set $R_1=M+0.3\sqrt{5}K$ and $R_2=-M+0.3\sqrt{5}K$. This family holds the number of actions fixed while varying the rank of the structured payoff component.

Following the \hyperref[par:evaluation-method] {evaluation method} described above, we fit both equilibrium models to each target and use $B=1000$ times of bootstrap. Complete inference details are provided in Appendix~\ref{app:inference}.

\noindent\textbf{Results.} Across the seven games and 147 model-pair targets, CQRE-pa is favored in 111 comparisons (75.5\%), Logit-QRE is favored in 24 (16.3\%), and 12 (8.2\%) are unresolved. CQRE-pa has lower mean KL in six of the seven games and sweeps all 21 pairings in \texttt{sm\_n5}, \texttt{sm\_n7}, and \texttt{rank1\_n5}. Complete game-level KL values, paired intervals, and target verdicts are reported in Appendix Table~\ref{tab:multigame-results}.

\begin{figure}[t]
    \centering

    \includegraphics[width=.48\columnwidth]
    {\detokenize{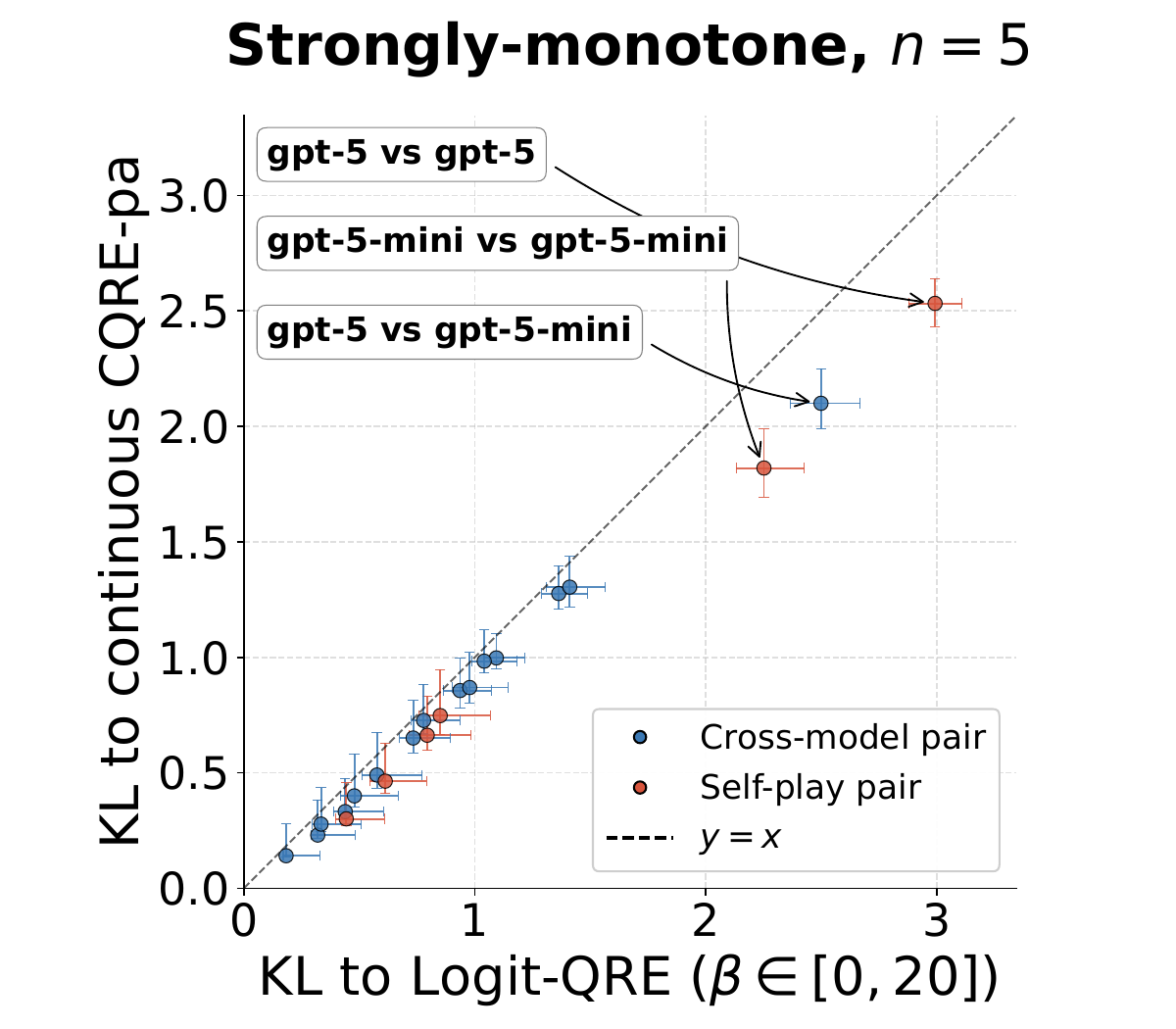}}
    \hfill
    \includegraphics[width=.48\columnwidth]
    {\detokenize{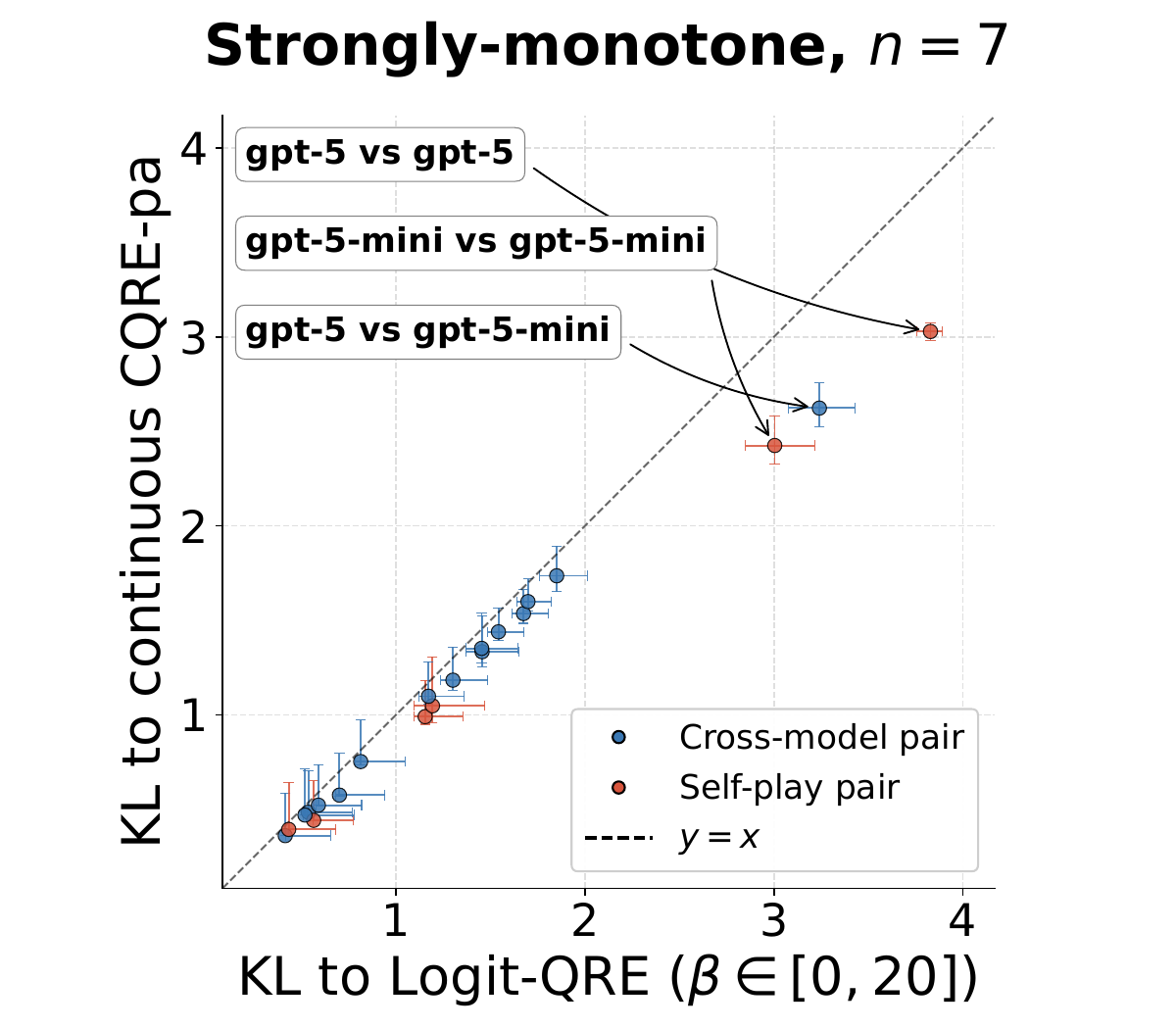}}

    \par\vspace{-1mm}

    \includegraphics[width=.48\columnwidth]
    {\detokenize{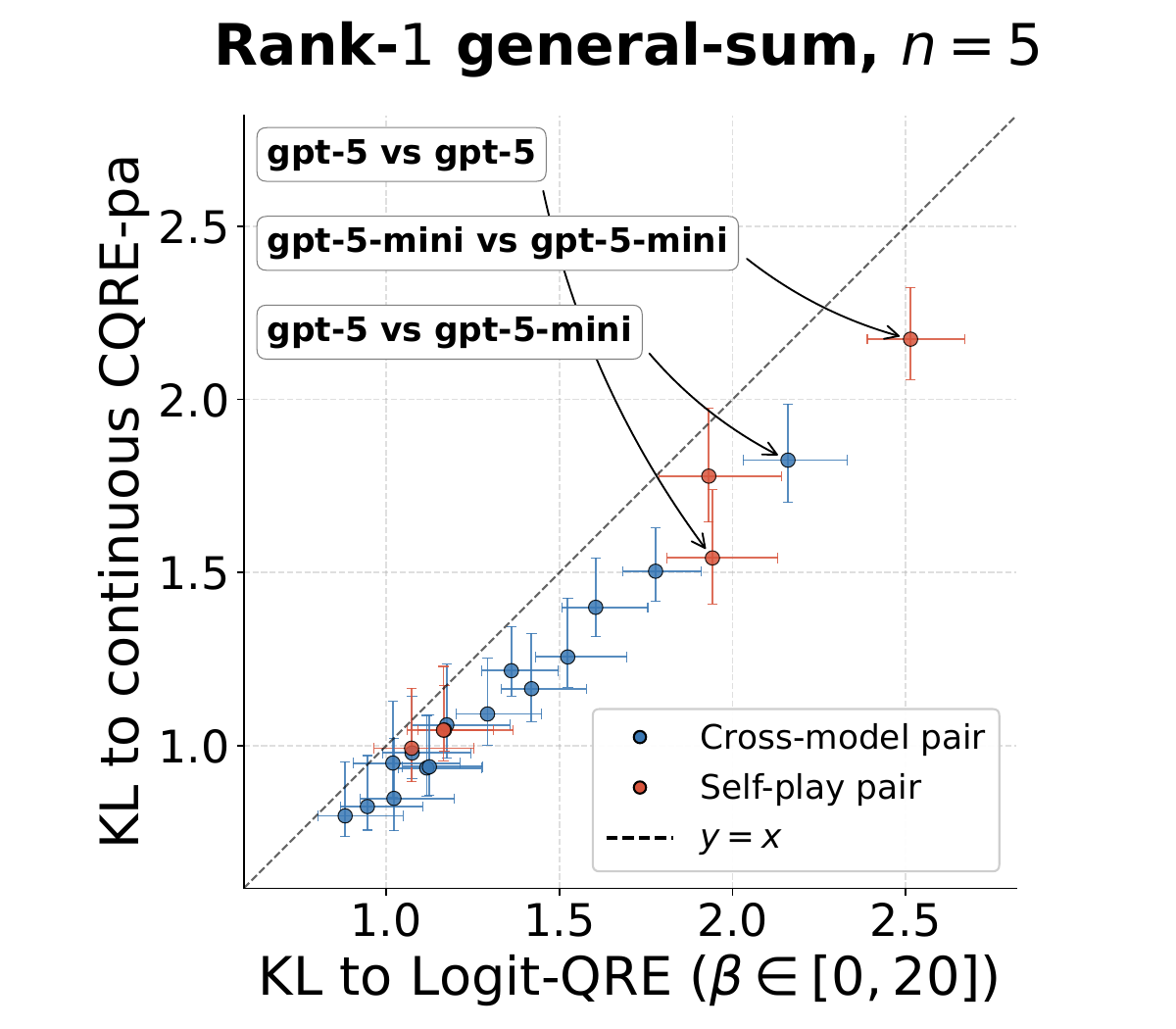}}
    \hfill
    \includegraphics[width=.48\columnwidth]
    {\detokenize{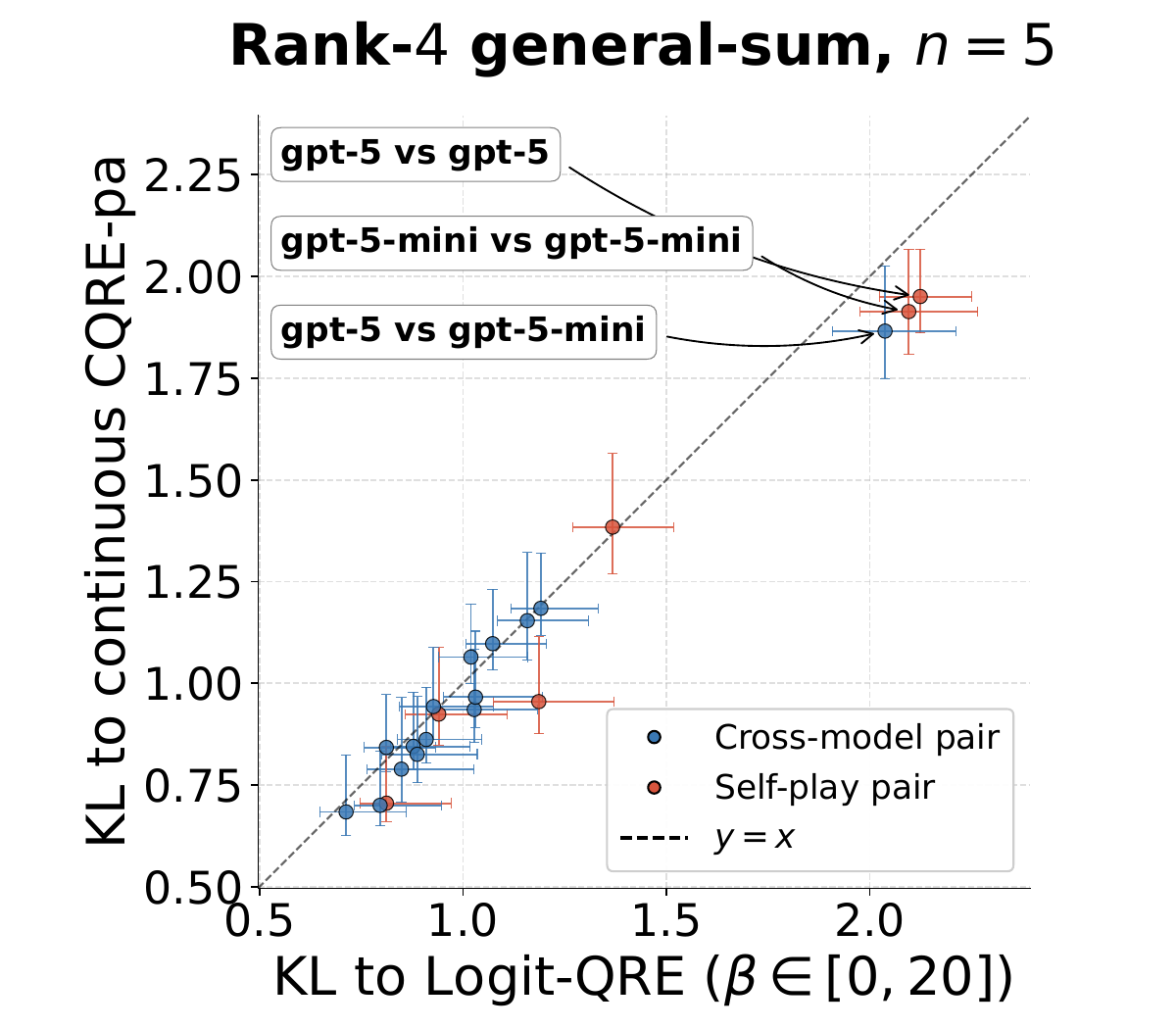}}

    \caption{Selected larger-action comparisons using CQRE-pa
    and Logit-QRE fits over $\beta\in[0,20]$. Each panel contains
    21 model pairings; points below the diagonal favor CQRE-pa.
    CQRE-pa sweeps \texttt{sm\_n5}, \texttt{sm\_n7}, and
    \texttt{rank1\_n5}, and has lower aggregate KL on
    \texttt{rank4\_n5}. Complete figures in
  Appendix Figure~\ref{fig:multigame-all}.}
    \label{fig:multigame-selected}
  \end{figure}

 CQRE-pa improves both strongly-monotone and rank-structured games, showing that its gains are not confined to $2{\times}2$ settings. Complete game- and target-level results are reported in Appendix~\ref{app:multigame-results}.

\noindent\textbf{Where Logit-QRE wins.}
Sixteen of the 24 strict Logit-QRE wins occur in \texttt{rank2\_n5}. GPT-5-family self/near-self play is more concentrated in every game: across these 21 targets, mean $\Delta=+0.25$, compared with $-2.78$ for the other 126 targets. In \texttt{sm\_n3}, three GPT-5-family targets place 98--99\% of their mass on the unique pure Nash equilibrium, allowing Logit-QRE to approach a point mass while CQRE-pa remains less concentrated. Across games, their mean CQRE-pa KL is 1.7--3.3 times the corresponding within-game mean for other targets, suggesting that highly concentrated play is a stress case for maximum-entropy selection. This pattern is not uniform, however: CQRE-pa wins all 12 GPT-5-family comparisons across \texttt{sm\_n5}, \texttt{sm\_n7}, \texttt{rank1\_n5}, and \texttt{rank4\_n5}. Further details are reported in Appendix~\ref{app:multigame-results}.



\noindent\textbf{Effective slack dimension.} In a separate strongly-monotone size sweep over $n=3, \ldots,10$, we stack the 21 componentwise-minimal fitted slack vectors for each game. Let $r_{90}$ denote the smallest number of singular directions explaining at least 90\% of their centered variation. Across the eight games, $r_{90}$ ranges from 2 to 7, corresponding to 22--50\% of the nominal $2n$-dimensional slack space. Complete spectra and fitted-rank results are reported in Appendix~\ref{app:slack-rank}. Sensitivity to extending the Logit-QRE domain to $\beta\in[0,\infty]$ is reported in Appendix~\ref{app:logit-sensitivity}.

\subsection{Three-Player Games}
  \label{sec:three-player}
Next, we test whether CQRE-pa's relative performance extends beyond two-player games. We consider \texttt{sm\_3p3a}, a three-player, three-action strongly-monotone-style game with a unique pure Nash equilibrium at $(A,C,C)$. The six AI models form all 56 model triplets with replacement, with 200 independent plays per triplet. Each empirical target is a distribution over 27 joint action profiles. Following the \hyperref[par:evaluation-method]{evaluation} above, we use $B=1000$ bootstraps and refit both CQRE-pa and Logit-QRE within every replicate.

CQRE-pa reduces mean KL by 21.6\% relative to Logit-QRE. Moreover, the target-specific paired interval lies below zero for every model triplet, yielding 56 CQRE-pa wins and no unresolved or Logit-QRE-favored comparisons. Thus, CQRE-pa's descriptive advantage is not confined to 2 player interaction: it exhibit an even stronger fit for joint multi player games; results for all 56 model triplets are reported in Appendix~\ref{app:three-player-results}.








\section{Conclusions}

In this paper, we introduced CQRE, a maximum-entropy model of joint play under per-player and per-action relaxation of correlated-equilibrium obedience, and proved that it yields a unique prediction for every fixed slack vector. We developed a continuous bi-level procedure to estimate these slacks from observed play. Empirically, CQRE attains lower mean KL  than Logit-QRE on established human and AI--AI play, and in six of seven larger-action games, and all three-player comparisons against logit-QRE. We also contribute a reproducible novel dataset of 201,600 one-shot AI decisions across all twelve Selten--Chmura games.

\newpage

\subsection*{AI use statement}
We used generative AI tools to assist with language editing, \LaTeX{} formatting, and figures of the paper. In addition, language models are objects of study in this work and we use LLM to generate the one-shot decisions comprising the AI datasets under the experimental protocol described in the paper. The authors reviewed and revised all AI-assisted text, all dataset generation and numerical data were cross-checked, and the plots are verified. 
Generative AI was not used to make final decisions about the paper's theoretical claims, empirical conclusions, or interpretation of results; other required disclosure categories are not applicable. We take responsibility for the complete content of this work.


\bibliography{refs}
\section*{Checklist}



\begin{enumerate}

  \item For all models and algorithms presented, check if you include:
  \begin{enumerate}
    \item A clear description of the mathematical setting, assumptions, algorithm, and/or model. [Yes] Sections~\ref{sec:cqre} and~\ref{sec:compute} define CQRE and the bi-level estimation procedure for the best-fit CQRE.
    \item An analysis of the properties and complexity (time, space, sample size) of any algorithm. [Yes] Section~\ref{sec:compute} reports the dimensions of the optimization problems, while Appendix~\ref{app:cqre-computation} reports solver, numerical safeguards, and limitations of the nonconvex outer optimization.
    \item (Optional) Anonymized source code, with specification of all dependencies, including external libraries. [Yes] An anonymized source code, dependency specifications, and reproduction instructions.
  \end{enumerate}

  \item For any theoretical claim, check if you include:
  \begin{enumerate}
    \item Statements of the full set of assumptions of all theoretical results. [Yes] 
    \item Complete proofs of all theoretical results. [Yes] 
    \item Clear explanations of any assumptions. [Yes] 
  \end{enumerate}

  \item For all figures and tables that present empirical results, check if you include:
  \begin{enumerate}
    \item The code, data, and instructions needed to reproduce the main experimental results (either in the supplemental material or as a URL). [Yes] 
    \item All the training details (e.g., data splits, hyperparameters, how they were chosen). [Yes] 
    \item A clear definition of the specific measure or statistics and error bars (e.g., with respect to the random seed after running experiments multiple times). [Yes] 
    \item A description of the computing infrastructure used. (e.g., type of GPUs, internal cluster, or cloud provider). [Not Applicable] 
  \end{enumerate}

  \item If you are using existing assets (e.g., code, data, models) or curating/releasing new assets, check if you include:
  \begin{enumerate}
    \item Citations of the creator If your work uses existing assets. [Yes]  
    \item The license information of the assets, if applicable. [Yes]
    \item New assets either in the supplemental material or as a URL, if applicable. [Yes] 
    \item Information about consent from data providers/curators. [Not Applicable]
    \item Discussion of sensible content if applicable, e.g., personally identifiable information or offensive content. [Not Applicable]
  \end{enumerate}

  \item If you used crowdsourcing or conducted research with human subjects, check if you include:
  \begin{enumerate}
    \item The full text of instructions given to participants and screenshots. [Not Applicable]
    \item Descriptions of potential participant risks, with links to Institutional Review Board (IRB) approvals if applicable. [Not Applicable]
    \item The estimated hourly wage paid to participants and the total amount spent on participant compensation. [Not Applicable]
  \end{enumerate}

\end{enumerate}

\clearpage
\appendix
\thispagestyle{empty}

\onecolumn













\appendix
\onecolumn
\aistatstitle{Supplementary Materials}

\section{Additional Theoretical Results}
\label{app:theory}

This appendix gives the full arguments behind the fixed-slack CQRE
properties and the regularity results used by the continuous outer
estimator.  We use the notation of Sections~\ref{sec:cqre}
and~\ref{sec:compute}.  In particular, $C$ is the obedience matrix,
$D_{\mathrm{pa}}$ replicates each player--action slack across all deviations
from that action, and
\[
K_\gamma
=\left\{\rho\in\Real_+^{|A|}:
\mathbf 1_{|A|}^{\mathsf T}\rho=1,\quad
C\rho\le D_{\mathrm{pa}}\gamma\right\}.
\]
All vector inequalities in this appendix are componentwise.

\subsection{Fixed-Slack Existence, Uniqueness, and Nesting}

\paragraph{Proof of Theorem~\ref{thm:exists}.}
By Nash's existence theorem~\citep{nash1950equilibrium}, every finite normal-form game has a mixed Nash equilibrium
$\sigma=(\sigma_i)_{i\in\mathcal N}$.  Let
$\bar\rho=\bigotimes_{i\in\mathcal N}\sigma_i$ be its product joint
distribution.  For any $i$, recommended action $a_i$, and deviation
$\widetilde a_i\ne a_i$,
\begin{align*}
R_{i,a_i\rightarrow\widetilde a_i}(\bar\rho)
&=\sigma_i(a_i)
\sum_{a_{-i}\in A_{-i}}
\left(\prod_{j\ne i}\sigma_j(a_j)\right)
\left[u_i(\widetilde a_i,a_{-i})-u_i(a_i,a_{-i})\right].
\end{align*}
If $\sigma_i(a_i)>0$, the bracketed expectation is nonpositive because
$a_i$ is a best response to $\sigma_{-i}$; if $\sigma_i(a_i)=0$, the
entire expression is zero.  Thus $\bar\rho\in K_0\subseteq K_\gamma$, so
$K_\gamma$ is nonempty.  It is the intersection of the probability
simplex with finitely many closed half-spaces and is therefore compact and
convex.  Shannon entropy is continuous and strictly concave on the simplex
under the convention $0\log 0=0$.  It consequently attains its maximum on
$K_\gamma$, and strict concavity makes that maximizer unique.

\paragraph{Proof of Proposition~\ref{prop:nesting}.}
A scalar slack is a per-player slack with $\epsilon_i=\epsilon$ for every
$i$.  A per-player slack is a per-action slack with
$\gamma_{i,a_i}=\epsilon_i$ for every $a_i\in A_i$.  Under either
embedding, every relaxed-obedience right-hand side is unchanged, and hence
so are the feasible set and its unique maximum-entropy element.  This proves
both inclusions. 

\subsection{Exponential-Family Characterization}
\label{app:exp-proof}

The indicator in the following result makes explicit which recommendation
indexes each obedience multiplier.  It also records the correct
characterization when the entropy maximizer is on a face of the simplex.

For fixed $\gamma$, define its attainable support by
\begin{equation}
\mathcal S_\gamma
:=\left\{x\in A:\text{there exists }q\in K_\gamma
\text{ with }q(x)>0\right\}.
\label{eq:attainable-support}
\end{equation}

\begin{proposition}[Support-aware exponential form]
\label{prop:appendix-support-form}
Let $\rho_\gamma^\star$ be the unique maximum-entropy element of
$K_\gamma$.  Then
$\operatorname{supp}(\rho_\gamma^\star)=\mathcal S_\gamma$.  Moreover,
there exist multipliers
$\lambda_{i,b_i,\widetilde b_i}\ge0$ such that, for every
$x\in\mathcal S_\gamma$,
\begin{equation}
\rho_\gamma^\star(x)
=\frac{1}{Z_{\mathcal S_\gamma}(\lambda)}
\exp\!\left\{
\sum_{(i,b_i,\widetilde b_i)\in\mathcal D}
\lambda_{i,b_i,\widetilde b_i}
\mathbf 1\{x_i=b_i\}
\left[u_i(b_i,x_{-i})-u_i(\widetilde b_i,x_{-i})\right]
\right\},
\label{eq:appendix-support-gibbs}
\end{equation}
where $Z_{\mathcal S_\gamma}(\lambda)$ normalizes over
$\mathcal S_\gamma$.  For $x\notin\mathcal S_\gamma$,
$\rho_\gamma^\star(x)=0$.  The multipliers obey
\begin{equation}
\lambda_{i,b_i,\widetilde b_i}
\left[
R_{i,b_i\rightarrow\widetilde b_i}(\rho_\gamma^\star)
-\gamma_{i,b_i}
\right]=0.
\label{eq:appendix-complementarity}
\end{equation}
The distribution is unique, although the multiplier vector need not be.
If $\mathcal S_\gamma=A$, Equation~\eqref{eq:appendix-support-gibbs}
reduces to Proposition~\ref{prop:exp}.
\end{proposition}

\begin{proof}
For each $x\in\mathcal S_\gamma$, choose
$q^x\in K_\gamma$ with $q^x(x)>0$.  Their average is feasible and assigns
positive probability to every profile in $\mathcal S_\gamma$.  Suppose
instead that $\rho_\gamma^\star(x)=0$ for some
$x\in\mathcal S_\gamma$.  Mixing $\rho_\gamma^\star$ with a feasible
distribution that assigns positive probability to $x$ preserves
feasibility.  As the mixing weight $t$ decreases to zero, the new
$x$-coordinate contributes a term of order $-t\log t$ to entropy, whereas
the changes on the positive coordinates are only of order $t$.  Entropy
therefore increases for all sufficiently small $t>0$, contradicting
optimality.  Hence the optimizer is positive on all of
$\mathcal S_\gamma$ and zero outside it by the definition of attainable
support.

Restrict the optimization problem to the face supported on
$\mathcal S_\gamma$.  Negative entropy is differentiable there at the
optimizer.  Because the feasible set is polyhedral, its first-order normal
cone is generated by the normalization equality and the active obedience
inequalities.  Consequently, there are
$\nu\in\Real$ and $\lambda\ge0$ for which stationarity gives, for every
$x\in\mathcal S_\gamma$,
\begin{equation}
\log\rho_\gamma^\star(x)+1+\nu
+\sum_{(i,b_i,\widetilde b_i)\in\mathcal D}
\lambda_{i,b_i,\widetilde b_i}
\mathbf 1\{x_i=b_i\}
\left[u_i(\widetilde b_i,x_{-i})-u_i(b_i,x_{-i})\right]
=0.
\label{eq:appendix-gibbs-stationarity}
\end{equation}
Solving for $\rho_\gamma^\star(x)$ and absorbing $1+\nu$ into the
normalizing constant yields Equation~\eqref{eq:appendix-support-gibbs}.
Dual feasibility and complementary slackness give
Equation~\eqref{eq:appendix-complementarity}.  Strict concavity of entropy
establishes uniqueness of the distribution.
\end{proof}

\subsection{Compact Slack Domain and Canonical Slacks}

Recall the maximum unilateral gain
\[
\overline d_{i,a_i}
=\max\!\left\{0,
\max_{\substack{\widetilde a_i\ne a_i\\a_{-i}\in A_{-i}}}
\left[u_i(\widetilde a_i,a_{-i})-u_i(a_i,a_{-i})\right]
\right\}
\]
and the box
$\Gamma_{\mathrm{pa}}
=\prod_{(i,a_i)\in\mathcal I}[0,\overline d_{i,a_i}]$.

\begin{proposition}[Clipping is prediction preserving]
\label{prop:appendix-clipping}
For any $\gamma\in\Real_+^{d_{\mathrm{pa}}}$, define
\[
\widetilde\gamma_{i,a_i}
:=\min\{\gamma_{i,a_i},\overline d_{i,a_i}\}.
\]
Then $\widetilde\gamma\in\Gamma_{\mathrm{pa}}$,
$K_{\widetilde\gamma}=K_\gamma$, and
$\rho_{\widetilde\gamma}^\star=\rho_\gamma^\star$.
Thus restricting the outer problem to $\Gamma_{\mathrm{pa}}$ removes no
attainable CQRE-pa prediction.
\end{proposition}

\begin{proof}
For any obedience row $d=(i,a_i,\widetilde a_i)$ and any
$\rho\in\Delta(A)$,
\begin{align*}
(C\rho)_d
&=\sum_{a_{-i}}\rho(a_i,a_{-i})
\left[u_i(\widetilde a_i,a_{-i})-u_i(a_i,a_{-i})\right]\\
&\le \overline d_{i,a_i}
\sum_{a_{-i}}\rho(a_i,a_{-i})
\le \overline d_{i,a_i}.
\end{align*}
Therefore, once the coordinate $\gamma_{i,a_i}$ reaches
$\overline d_{i,a_i}$, every obedience inequality replicated from that
coordinate is satisfied by every distribution in the simplex.  Replacing a
larger coordinate by $\overline d_{i,a_i}$ consequently leaves the feasible
set unchanged.  Applying this argument to every coordinate proves the
claim, and equality of the unique inner solutions follows immediately.
\end{proof}

A useful prediction-preserving representative of outer slack vector is obtained directly from a fitted
joint distribution.  Define
\begin{equation}
\gamma^{\mathrm{can}}_{i,a_i}(\rho)
:=\max\!\left\{0,
\max_{\widetilde a_i\ne a_i}
R_{i,a_i\rightarrow\widetilde a_i}(\rho)
\right\}.
\label{eq:canonical-slack}
\end{equation}

\begin{proposition}[Componentwise-minimal representative]
\label{prop:canonical-slack}
If $\rho=\rho_\gamma^\star$, then
$\rho=\rho_{\gamma^{\mathrm{can}}(\rho)}^\star$.
Moreover, $\gamma^{\mathrm{can}}(\rho)$ is the componentwise-smallest
nonnegative slack vector for which $\rho$ is feasible.
\end{proposition}

\begin{proof}
Feasibility of $\rho$ under $\gamma$ implies
$\gamma^{\mathrm{can}}(\rho)\le\gamma$.  By construction,
$\rho\in K_{\gamma^{\mathrm{can}}(\rho)}$, while
$K_{\gamma^{\mathrm{can}}(\rho)}\subseteq K_\gamma$.  Since $\rho$ is
the unique maximum-entropy distribution on the larger set, it remains the
unique maximum-entropy distribution on the smaller set containing it.  If
$\rho\in K_\eta$ for any other nonnegative $\eta$, each obedience
inequality requires
$\eta_{i,a_i}\ge
R_{i,a_i\rightarrow\widetilde a_i}(\rho)$, and nonnegativity then gives
$\eta\ge\gamma^{\mathrm{can}}(\rho)$.
\end{proof}

\subsection{Continuity and Existence of an Outer Minimizer}
\label{app:outer-theory}

The global outer objective need not be differentiable everywhere, and its
minimizing slack vector need not be unique.  The following result gives the
existence statement required by Equation~\eqref{eq:outer-problem} without
asserting either property.

\begin{proposition}[Continuity of the inner prediction and outer existence]
\label{prop:appendix-outer-existence}
The correspondence $\gamma\mapsto K_\gamma$ is nonempty,
compact-valued, and continuous on $\Gamma_{\mathrm{pa}}$.  The unique
inner prediction $\gamma\mapsto\rho_\gamma^\star$ is continuous.  For
fixed $\hat p\in\Delta(A)$, the extended-real loss
\[
L_{\mathrm{pa}}(\gamma)
=\KL(\hat p\Vert\rho_\gamma^\star)
\in[0,+\infty]
\]
is lower semicontinuous on $\Gamma_{\mathrm{pa}}$ and attains a finite
minimum there.  The set of outer minimizers is therefore nonempty and
compact, but it need not be a singleton.
\end{proposition}

\begin{proof}
Nonemptiness was established in the proof of Theorem~\ref{thm:exists}, and
compactness follows because $K_\gamma$ is closed in the simplex.  The graph
of the correspondence is closed: if
$\gamma^n\to\gamma$, $\rho^n\to\rho$, and
$\rho^n\in K_{\gamma^n}$, passage to the limit in the linear constraints
gives $\rho\in K_\gamma$.  Compactness of the common ambient simplex then
gives upper hemicontinuity.

For lower hemicontinuity, fix $\rho\in K_\gamma$ and choose any
$q\in K_0$.  Given $\gamma^n\to\gamma$, define
\[
t_n:=\max_{g:\,\gamma_g>0}
\left(1-\frac{\gamma_g^n}{\gamma_g}\right)_+,
\]
with $t_n=0$ if every coordinate of $\gamma$ is zero.  Then $t_n\to0$
and $(1-t_n)\gamma\le\gamma^n$.  The mixture
$\rho^n=(1-t_n)\rho+t_nq$ converges to $\rho$ and satisfies
\[
C\rho^n
\le (1-t_n)D_{\mathrm{pa}}\gamma
\le D_{\mathrm{pa}}\gamma^n,
\]
so $\rho^n\in K_{\gamma^n}$.  Hence the feasible-set correspondence is
also lower hemicontinuous.

Negative entropy is continuous, and its minimizer over each $K_\gamma$ is
unique. Berge's maximum theorem~\citep{berge1963espaces} therefore implies continuity of
$\rho_\gamma^\star$.  For fixed $\hat p$, KL divergence is lower
semicontinuous in its second argument when assigned value $+\infty$ if
$\hat p(x)>0$ but the second argument is zero.  Its composition with the
continuous inner prediction is consequently lower semicontinuous.  At the
upper corner $\gamma_{i,a_i}=\overline d_{i,a_i}$, all obedience
constraints are redundant, so the inner prediction is uniform and the KL
loss is finite.  A finite minimum is therefore attained on the compact box
$\Gamma_{\mathrm{pa}}$.  Lower semicontinuity also makes the argmin set
closed, hence compact. 
\end{proof}

\subsection{Implicit Differentiation of the Inner KKT System}
\label{app:implicit-gradient}

We next give the local derivative used by the continuous optimizer.  Fix
$\gamma\in\Gamma_{\mathrm{pa}}$, write
$\rho^\star:=\rho_\gamma^\star$, and let
\[
S:=\left\{d\in\mathcal D:
(C\rho^\star)_d=(D_{\mathrm{pa}}\gamma)_d\right\}
\]
be the binding obedience rows.  Let $A_S:=C_{S,:}$ and
$D_S:=(D_{\mathrm{pa}})_{S,:}$ denote the rows indexed by $S$ of the
main-text obedience matrix $C$ and $D_{\mathrm{pa}}$, respectively, and define
\begin{equation}
C_S:=
\begin{bmatrix}
\mathbf 1_{|A|}^{\mathsf T}\\ A_S
\end{bmatrix},
\qquad
G_S:=
\begin{bmatrix}
0_{1\times d_{\mathrm{pa}}}\\ D_S
\end{bmatrix},
\qquad
W:=\operatorname{diag}(\rho^\star).
\label{eq:appendix-kkt-matrices}
\end{equation}

\begin{proposition}[Local CQRE-pa hypergradient]
\label{prop:appendix-implicit-gradient}
Assume that $\gamma$ is interior with respect to every nondegenerate slack coordinate and that (i) $\rho^\star$ has full support; (ii) the inactive obedience inequalities are strict and the active multipliers are positive; (iii) $C_S$ has full row rank; and (iv) the same active set remains valid locally.  Then the inner prediction is differentiable in a neighborhood of $\gamma$, with
\begin{equation}
J_\gamma
:=\frac{\partial\rho_\gamma^\star}{\partial\gamma}
=WC_S^{\mathsf T}
\left(C_S WC_S^{\mathsf T}\right)^{-1}
G_S.
\label{eq:appendix-implicit-jacobian}
\end{equation}
If $\rho^\star(x)>0$ for every $x$ with $\hat p(x)>0$, then
\begin{equation}
\nabla_\gamma L_{\mathrm{pa}}(\gamma)
=-J_\gamma^{\mathsf T}
\left(\hat p\oslash\rho^\star\right).
\label{eq:appendix-outer-gradient}
\end{equation}
These formulas are local: they do not imply a closed form or a unique
solution for the outer optimization problem.
\end{proposition}

\begin{proof}
Let $y=(\nu,\lambda_S^{\mathsf T})^{\mathsf T}$ collect the multiplier
of the normalization equality and the positive multipliers of the binding
obedience constraints.  Full support removes the nonnegativity constraints
from the local stationarity system.  The KKT equations are
\begin{align}
\log\rho^\star+\mathbf 1_{|A|}+C_S^{\mathsf T}y&=0,
\label{eq:appendix-local-stationarity}\\
C_S\rho^\star
&=\begin{bmatrix}1\\D_S\gamma\end{bmatrix}.
\label{eq:appendix-local-feasibility}
\end{align}
Their Jacobian with respect to $(\rho,y)$ is
\[
\mathcal K_S=
\begin{bmatrix}
W^{-1}&C_S^{\mathsf T}\\
C_S&0
\end{bmatrix}.
\]
This matrix is nonsingular.  Indeed, if
$W^{-1}u+C_S^{\mathsf T}v=0$ and
$C_Su=0$, multiplying the first equality by $u^{\mathsf T}$
gives $u^{\mathsf T}W^{-1}u=0$.  Positive definiteness of $W^{-1}$
implies $u=0$; full row rank of $C_S$ then implies $v=0$.
The implicit-function theorem therefore applies.  Strict inactive
inequalities and positive active multipliers are the strict-complementarity
conditions that keep the same active set locally.

Let $Z:=\partial y/\partial\gamma$.  Differentiating
Equations~\eqref{eq:appendix-local-stationarity}--
\eqref{eq:appendix-local-feasibility} gives
\[
W^{-1}J_\gamma+C_S^{\mathsf T}Z=0,
\qquad
C_SJ_\gamma=G_S.
\]
The first equality gives
$J_\gamma=-WC_S^{\mathsf T}Z$.  Substitution into the second
gives
\[
-C_SWC_S^{\mathsf T}Z=G_S.
\]
Because $W$ is positive definite and $C_S$ has full row rank,
$C_SWC_S^{\mathsf T}$ is positive definite and
invertible.  Solving for $Z$ and substituting back proves
Equation~\eqref{eq:appendix-implicit-jacobian}.  Finally,
\[
\frac{\partial}{\partial\rho}
\KL(\hat p\Vert\rho)=-\hat p\oslash\rho,
\]
so the chain rule gives Equation~\eqref{eq:appendix-outer-gradient}.
\end{proof}

At active-set changes, boundary-support changes, weakly active constraints,
or failures of the rank and conditioning assumptions, the outer loss may
fail to be differentiable or Equation~\eqref{eq:appendix-implicit-jacobian}
may be numerically unreliable.  The safeguarded finite-difference fallback
used in those cases is described in Appendix~\ref{app:cqre-computation}.

\section{Computational Details}
\label{app:computation}
\label{app:cqre-computation}
\label{app:cqre-numerics}

This appendix specifies the numerical procedures used to compute CQRE for a
fixed slack vector and to estimate the slack vector from an empirical joint
distribution.  It also defines the equilibrium benchmarks used in the
experiments.  The statistical resampling procedures are described separately
in Appendix~\ref{app:inference}.

\subsection{Inner maximum-entropy solve}
\label{app:inner-solver}

For a fixed per-action slack vector $\gamma$, the numerical inner problem is
the matrix program in Equation~\eqref{eq:inner-matrix-program}.  We minimize
negative Shannon entropy using sequential least-squares programming
(SLSQP; \citealp{kraft2005slsqp}) through
\texttt{scipy.optimize.minimize}.  The implementation supplies the objective
gradient $1+\log\rho$ and imposes the normalization and relaxed-obedience
constraints directly.

To keep the logarithm numerically defined, the implementation uses the lower
bound
\[
 \rho(x)\ge 10^{-12}\qquad(x\in A).
\]
This is a numerical floor, not part of Definition~\ref{def:cqre}.  The standard
inner-solver settings are \texttt{maxiter=500} and
\texttt{ftol=$10^{-12}$}.  A cold solve begins from the uniform joint
distribution.  During continuous outer optimization, the inner solver is
warm-started from the most recent feasible joint distribution; the retained
grid incumbent is available as a fallback start.

Every returned solution is checked for finite entries, nonnegativity,
unit total mass, and satisfaction of the obedience constraints.  The default
two-player feasibility threshold is $10^{-7}$.  The three-player production
verifier uses a $2\times10^{-7}$ constraint threshold and a
$2\times10^{-8}$ probability/normalization threshold.  A candidate that does
not pass the applicable checks cannot replace an incumbent.  The final KL is
always recomputed from the selected joint distribution.

\subsection{Outer domains and discrete incumbent}
\label{app:outer-domain-computation}

For each player--action pair define
\[
 \overline d_{i,a_i}
 :=\max\!\left\{0,
 \max_{\widetilde a_i\ne a_i}\max_{a_{-i}}
 \bigl[u_i(\widetilde a_i,a_{-i})-u_i(a_i,a_{-i})\bigr]
 \right\}.
\]
 As shown in Proposition~\ref{prop:appendix-clipping}, clipping a per-action slack to
$\overline d_{i,a_i}$ cannot change the feasible set.  The continuous per-action domain
is therefore
\[
 \Gamma_{\mathrm{pa}}
 =\prod_{i,a_i}[0,\overline d_{i,a_i}].
\]
The corresponding complete domains for the restricted specifications are
\[
 \Gamma_{\mathrm{pp}}
 =\prod_i\left[0,\max_{a_i}\overline d_{i,a_i}\right],
 \qquad
 \Gamma_{\mathrm{sc}}
 =\left[0,\max_{i,a_i}\overline d_{i,a_i}\right].
\]
After the normalization above, every upper endpoint is at most one.

Before continuous refinement, we construct a discrete incumbent using
\[
 \mathcal G_0=\{0,0.005,0.01,0.02,0.05,0.1,0.2\}.
\]
Scalar and per-player searches enumerate the relevant grid whenever this is
tractable.  For a two-player $2\times2$ game, exhaustive per-action search has
$7^4=2401$ candidates.  In larger parameter spaces, the historical fitter
uses deterministic coordinate-wise grid search.  This grid solution is an initializer and a hard
incumbent; it is not the final restriction on the continuous estimator.

\subsection{Continuous outer optimization}
\label{app:continuous-outer-computation}

For each slack parameterization, the outer objective is
\[
 L_{\mathrm{pa}}(\gamma)=\KL\!\left(\hat p\,\middle\|\,
 \rho_\gamma^\star\right).
\]
We optimize this objective with box-constrained, multistart L-BFGS-B~\citep{byrd1995limited}.  The
start bank contains the discrete incumbent and feasible deterministic starts;
for a finer slack parameterization it also contains the fitted solutions
embedded from the coarser parameterizations.  All objective evaluations solve
the inner maximum-entropy problem first.  Thus the procedure never replaces
the bilevel problem by a joint optimization over $(\rho,\gamma)$.

The default continuous refiner uses \texttt{maxiter=100},
\texttt{ftol=$10^{-12}$}, \texttt{gtol=$10^{-6}$}, and
\texttt{maxls=20}.  Its default feasibility tolerance is $10^{-7}$.
Panel-specific production jobs may increase the iteration budget or add
independent rescue starts, but do not change the objective or the outer box.
The exact run manifests and software environment are included with the
reproducibility artifact.

\paragraph{Analytic-gradient eligibility.}
Let $S$ be the active obedience constraints at the current inner solution.
  Let $A_S:=C_{S,:}$ and
  $D_S:=(D_{\mathrm{pa}})_{S,:}$ denote the rows indexed by $S$, and define
  \[
   C_S=
   \begin{bmatrix}
   \mathbf 1_{|A|}^{\mathsf T}\\
   A_S
   \end{bmatrix},
   \qquad
   G_S=
   \begin{bmatrix}
   0_{1\times d_{\mathrm{pa}}}\\
   D_S
   \end{bmatrix},
   \qquad
   W=\operatorname{diag}(\rho_\gamma^\star).
  \]
On a cell with an unchanged active set, full support, strict
complementarity, and full row rank of $C_S$, implicit differentiation of the
inner KKT system gives
\[
 J_\gamma
 =\frac{\partial\rho_\gamma^\star}{\partial\gamma}
 =WC_S^{\mathsf T}
  (C_SWC_S^{\mathsf T})^{-1}G_S,
\]
and hence the gradient in Equation~\eqref{eq:outer-gradient}.  The full
derivation is given in Proposition~\ref{prop:appendix-implicit-gradient}.

The implementation uses this derivative only after checking primal
feasibility, stationarity, multiplier signs, complementarity, strictness of
inactive constraints, full support, active-row rank, numerical conditioning,
and the residual of the differentiated equality system. 

\paragraph{Finite-difference fallback.}
If any analytic-gradient check fails, we use bounded centered finite
differences, with one-sided differences at a box boundary and zero derivative
for a fixed coordinate.  The baseline relative step is
$10^{-4}\max\{1,|\gamma_j|\}$.  The audited larger-action production pipeline
also runs an independent four-start finite-difference lane with step
$1.25\times10^{-5}$ and at most 360 iterations, retaining whichever eligible
lane has lower recomputed KL.  Derivative-only probe points are not silently
selected: a materially improving feasible probe must first seed an explicit
optimization restart and pass the same candidate checks.

\paragraph{Incumbent and nesting guarantees.}
The final selector compares all eligible continuous candidates against the
unchanged grid fit.  It returns a continuous candidate only if its recomputed
KL is no larger than the incumbent up to the pipeline's stated numerical
tolerance.  Embedded coarser solutions are likewise retained when fitting a
finer parameterization.  Consequently, up to numerical tolerance,
\[
 L_{\mathrm{pa}}^{\mathrm{selected}}
 \le L_{\mathrm{pp}}^{\mathrm{selected}}
 \le L_{\mathrm{sc}}^{\mathrm{selected}},
 \qquad
 L_m^{\mathrm{selected}}\le L_m^{\mathrm{grid}}.
\]
Because the outer loss can have kinks, flat directions, and multiple local
minima, these guarantees are incumbent guarantees, not a proof of global
optimality of L-BFGS-B.

\paragraph{Canonical slack representation.}
The inner joint distribution is unique for fixed slacks, but a joint
distribution can have more than one slack representation.  For a fitted
joint $\rho$, we use the componentwise-minimal feasible representative
\[
 \gamma_{i,a_i}^{\mathrm{can}}(\rho)
 =\max\!\left\{0,
 \max_{\widetilde a_i\ne a_i}
 R_{i,a_i\rightarrow\widetilde a_i}(\rho)
 \right\}.
\]
If $\rho=\rho_\gamma^\star$, replacing $\gamma$ by
$\gamma^{\mathrm{can}}(\rho)$ preserves $\rho$: the new feasible set is a
subset of the old one that still contains $\rho$, so the old unique
maximum-entropy optimizer remains optimal.  We use canonical rather than raw
optimizer endpoints in the effective-slack-dimension analysis.

\subsection{Equilibrium benchmarks}
\label{app:benchmark-computation}

All benchmarks are scored with the same direction of KL,
$\KL(\hat p\Vert\rho)$.

\paragraph{Restricted and exact CQRE.}
Scalar CQRE and per-player CQRE use the same inner solver and continuous
outer procedure as CQRE-pa, with the replication restrictions defined in
Section~\ref{sec:cqre}.  Exact CQRE fixes $\gamma=0$ and therefore requires
only one inner maximum-entropy solve.

\paragraph{Mixed Nash.}
For the finite games in which mixed Nash is reported, equilibria are obtained
by support enumeration and validated by best-response inequalities.  If more
than one validated equilibrium is found, the table reports the product joint
distribution with the smallest in-sample KL to the target.  This is an
oracle selection among equilibria and is labeled as such; it is not a unique
equilibrium-selection rule implied by Nash equilibrium itself.

\subsection{Logit-QRE computation}
\label{app:logit-computation}
\paragraph{Principal branch.}
  For a fixed value of $\beta$, the Logit-QRE equations may have
  more than one solution. We therefore begin at $\beta=0$, where
  the unique solution is uniform play, and gradually increase
  $\beta$, using each solution to initialize the next solve. The
  continuous family of solutions reached in this way is called the
  \emph{principal branch}. For each empirical target, we select the
  point on this branch that minimizes KL over the specified
  $\beta$ domain. When Logit-QRE is unique, this is equivalent to
  ordinary one-dimensional optimization over $\beta$. For fits on
  $[0,20]$, both endpoints are explicitly considered. The
  three-player implementation instead solves from a uniform
  initialization at each candidate $\beta$ and verifies the
  fixed-point residual of the selected solution.

  \paragraph{Including $\beta=\infty$.}
  The value $\beta=\infty$ represents the limiting mixed-Nash
  prediction approached as $\beta$ grows along the principal
  branch. It is not approximated using one arbitrarily chosen
  large finite value. For the branch-audited two-player analysis,
  we continue the branch beyond $\beta=20$ and test whether its
  solutions converge toward a validated mixed-Nash equilibrium.
  The limiting equilibrium is included as the
  $\beta=\infty$ candidate only when the required distance and
  convergence checks pass. The best candidate from $[0,20]$ is
  always retained, so extending the domain to $[0,\infty]$ cannot
  worsen the fitted Logit-QRE KL. If the limiting endpoint cannot
  be validated, it is not included.

\section{Data and Experimental Protocol}
\label{app:data-protocol}
\label{app:data}

This appendix documents the four empirical panels used in the main paper:
the Selten--Chmura human benchmark, the all-12 human--human/AI--AI/human--AI
extension, the seven larger-action two-player games, and the
three-player \texttt{sm\_3p3a} game. 

\subsection{Collection interface and common LLM protocol}

The three AI panels use six logical OpenAI deployments: \texttt{gpt-4o}, \texttt{gpt-4o-mini}, \texttt{gpt-4.1}, \texttt{gpt-4.1-mini}, \texttt{gpt-5}, and \texttt{gpt-5-mini}~\citep{openai2026models}.  Each request is a stateless one-shot decision: the prompt states that play is simultaneous, supplies no history or future interaction, does not identify the other deployment, and contains no public signal.  The common system instruction is to return one uppercase action label without an explanation.  The two-player prompt displays the full payoff bimatrix and identifies whether the respondent is the row or column player.  The three-player prompt identifies the respondent's player position and displays that player's payoff for every profile of the other players; it does not display the other players' payoff functions.

For the all-12 and seven-game collections, the requested sampling settings for deployments were temperature $0.8$, top-$p$ $0.95$, and a four-token visible-output budget.   Local seeds controlled the task schedule, roles, and label permutations; the all-12 collection did not send a provider-side sampling seed.  Raw returned text and the parsed action
were both retained.  Empty or unparseable responses were retried under the frozen acceptance rule, with stable task identifiers preventing an accepted decision from being collected twice.

\subsection{Target definitions and counts}
\label{app:targets}

We call a \emph{target} one empirical joint-action distribution fitted
separately by each equilibrium model.  Table~\ref{tab:protocol-targets}
summarizes the target universe.  Target counts are derived evaluation cells;
they are not counts of independent observational units.

\begin{table*}[t]
\centering
\small
\caption{Target construction in the four empirical panels.  The primitive
unit column identifies the data from which each target is formed.}
\label{tab:protocol-targets}
\begin{tabular}{@{}p{0.18\textwidth}p{0.23\textwidth}rp{0.42\textwidth}@{}}
\hline
Panel & Target index & Targets & Primitive data and target construction \\
\hline
Selten--Chmura humans
& game
& 12
& Equal-weight mean of the released within-group joint distributions for one
  game. \\
All-12 extension
& game $\times$ ordered row agent $\times$ ordered column agent
& 588
& 432 AAI paired/reconstructed joints, 144 directional HAI marginal products,
  and 12 HH targets from released human joint distributions. \\
Seven larger-action games
& game $\times$ unordered model dyad
& 147
& Seven games $\times$ 21 model dyads; each target contains 200 one-shot
  action pairs. \\
\texttt{sm\_3p3a}
& unordered model multiset
& 56
& Every size-three multiset from the six deployments; each target contains
  200 one-shot action triples. \\
\hline
\end{tabular}
\end{table*}

\paragraph{Selten--Chmura human benchmark.}
We use the released replication workbook for the Selten--Chmura experiment~\citep{selten2008stationary,selten2008replication}. It contains twelve $2\times2$ games.  Games \texttt{SC1}--\texttt{SC6} have 12 independent subject
groups each, while \texttt{SC7}--\texttt{SC12} have 6 each, for 108 groups in
total.  A group contains four fixed-role row subjects and four fixed-role
column subjects who are randomly matched over 200 periods, producing four
realized action pairs per period and 800 pairs in total.  The workbook
releases each group's four-cell joint frequencies for periods 1--100,
101--200, and 1--200.  The main analysis uses periods 1--200.   The twelve games then receive equal weight in the headline
aggregate.

\paragraph{All-12 HH, AAI, and constructed HAI targets.}
The same twelve presentations form six payoff-translation-equivalent
strategic families, \texttt{SC1}/\texttt{SC7} through
\texttt{SC6}/\texttt{SC12}.  The analysis treats the six AI systems and the
published Selten--Chmura human population as seven agents and constructs all $7\times7=49$
ordered row-agent--column-agent cells in every game.  This gives
$12\times49=588$ targets: $12\times36=432$ AI--AI (AAI),
$12\times12=144$ directional human--AI (HAI), and 12 human--human (HH).
The HAI targets split evenly between human-row/AI-column and
AI-row/human-column orientations.

Each AAI target contains 200 one-shot action pairs. Some pairs contain two decisions collected for the same match; the remaining pairs are formed from independently collected opponent-blind decisions. Each HAI target is the outer product of a model-role-specific 200-decision AI marginal and the corresponding human marginal. HAI is therefore constructed instead of observed human--AI interaction. HH uses the released within-group human joint distributions.

The all-12 Selten--Chmura AI dataset contains 201,600 accepted one-shot decisions from six model deployments across twelve games. Preflight and pilot calls, failed requests, and retried requests are excluded from this count. Each scheduled task contributes one accepted response. Collection validation checks all scheduled tasks, all 588 target cells, action-permutation balance, and accepted task identifiers.

\paragraph{Seven larger-action games.}
The panel combines three two-player strongly-monotone constructions with
$n\in\{3,5,7\}$ actions and four rank-$r$ general-sum constructions with
$n=5$ and $r\in\{1,2,3,4\}$.  For each game, the six deployments generate the
21 unordered dyads with replacement.  Every dyad contributes
200 independent one-shot matches, yielding 147 targets, 29,400 action pairs,
and 58,800 AI decisions.

\paragraph{Three-player \texttt{sm\_3p3a}.}
This panel contains one fixed three-player, three-action payoff tensor.  The
six deployments generate all 56 size-three model multisets with
replacement.  Each multiset contributes 200 one-shot matches, for 11,200
action triples and 33,600 AI decisions.  The tuple is stored in a canonical
player order and is not rotated across player positions; the target is thus
indexed by that recorded assignment, even though the set of deployment names
is enumerated as a multiset.

\subsection{Game construction }
\label{app:game-construction}

The Selten--Chmura panels use the twelve payoff tables from the released
workbook.  In the larger-action panel, let $K$ denote the normalized
antisymmetric draw.  The strongly-monotone construction is
$R_1=0.2I_n+\sqrt{n}K$ and $R_2=0.2I_n-\sqrt{n}K$.  The rank-$r$ generator draws
fixed $U,V\in\mathbb{R}^{5\times r}$, forms $M=UV^\top$, and sets
$R_1=M+0.3\sqrt{5}K$ and $R_2=-M+0.3\sqrt{5}K$.  The committed seeds are 7,
37, 127, and 50 for $r=1,2,3,4$, respectively; the strongly-monotone games use
seed 2.  Each pair of payoff matrices is then jointly rescaled to maximum
absolute entry 5 and rounded to one decimal place before prompting.

The \texttt{sm\_3p3a} tensor uses a small player-specific diagonal preference
($\mu=0.2$) plus three fixed pairwise antisymmetric couplings (scale 1, seed
11), followed by player-wise rescaling and rounding to one decimal place.  The
identifier is retained for continuity with the experiment.

\subsection{Bootstrap and inference}
  \label{app:bootstrap}
  \label{app:inference}

  For target $t$, we report
  \[
  \Delta_t
  =
  \KL\!\left(\hat p_t\middle\|
  \widehat\rho_t^{\mathrm{CQRE\text{-}pa}}\right)
  -
  \KL\!\left(\hat p_t\middle\|
  \widehat\rho_t^{\mathrm{Logit\text{-}QRE}}\right),
  \]
  so negative values favor CQRE-pa. Paired percentile intervals use the same resampled target for both models. The Selten--Chmura human analysis uses common-precision Logit-QRE over $\beta\in[0,20]$, the all-12 analysis uses $\beta\in[0,\infty]$, and the larger-action and \texttt{sm\_3p3a} analyses use $\beta\in[0,20]$; $\beta=\infty$ denotes the limiting mixed-Nash prediction. Intervals wholly below zero favor CQRE-pa, confidence intervals wholly above zero favor Logit-QRE, and the remaining comparisons are unresolved.

\paragraph{Panel-specific resampling.}
For the Selten--Chmura human benchmark, we use $B=2000$ bootstrap replicates. Within each game, subject groups are sampled with replacement, and 800 action pairs are then drawn from each selected group's reported distribution. Both models are refitted to the resulting target; groups receive equal weight within games, and the twelve games receive equal weight in the decisions and human observations before reconstructing their product distribution; and for HH targets, we use the same subject-group procedure as above. Within each replicate, targets that depend on the same human data use the same resampled human observations. These intervals therefore exclude uncertainty from refitting the models. For the three-player analysis, we use $B=1000$ replicates, resample the 200 action triples within each target, reconstruct the 27-cell distribution, and refit both models. All reported comparisons measure descriptive fit to the observed targets instead of out-of-sample predictive performance.

\subsection{Data, Code, and Licenses}
  \label{app:asset-licenses}

  \paragraph{Asset licenses.}
The Selten--Chmura replication data are available from openICPSR (Project 113254; DOI: \url{https://doi.org/10.3886/E113254V1}). Its data are licensed under CC BY 4.0 and its accompanying software under the BSD 3-Clause license.
\section{Additional Empirical Results}
\label{app:additional-results}
\label{app:results}

This appendix reports additional results supporting
Section~\ref{sec:empirical}.  Throughout, we use
\begin{equation}
  \Delta
  :=\KL(\hat p\Vert\widehat\rho_{\mathrm{CQRE\text{-}pa}})
  -\KL(\hat p\Vert\widehat\rho_{\mathrm{Logit\text{-}QRE}}),
  \label{eq:appendix-delta}
\end{equation}
so negative values favor CQRE-pa.  Unless stated otherwise, reported
intervals are percentile intervals from synchronized bootstrap
replicates, and aggregate quantities weight targets equally.

\subsection{Established Human Benchmark}
  \label{app:selten-results}

For each Selten--Chmura game, the empirical target is the equal-weight average of its observed group-level joint distributions. CQRE-pa achieves the lowest point KL in nine of the twelve games and reduces the equal-game mean KL from $.0093$ for Logit-QRE to $.0040$, a 56.9\% reduction. Fit also improves across the nested CQRE specifications, from $.0072$ for common-slack CQRE to $.0057$ for per-player CQRE and $.0040$ for CQRE-pa. The aggregate paired interval for $\Delta$ is strictly below zero, although the game-level results remain heterogeneous. Table~\ref{tab:selten} reports the complete benchmark comparison, and Figure~\ref{fig:selten-human-kl} visualizes the CQRE-pa and Logit-QRE fits.

 \begin{table}[t]
    \centering
    \scriptsize
    \caption{Fit to the twelve Selten--Chmura human games.
    N/E denotes the identical mixed-Nash and exact-CQRE
    entries; sc., pp., and pa denote scalar, per-player,
    and per-action CQRE. The first five model columns
    report full-sample KL divergence. The final column
    reports
    $\Delta=\KL_{\mathrm{CQRE\text{-}pa}}
    -\KL_{\mathrm{Logit\text{-}QRE}}$
    with paired 95\% bootstrap intervals. Lower KL is
    better; negative $\Delta$ favors CQRE-pa.}
    \label{tab:selten}

    {\setlength{\tabcolsep}{1.2pt}
    \begin{tabular}{@{}lrrrrrr@{}}
    \hline
    & \multicolumn{5}{c}{Full-sample KL divergence}
    & \multicolumn{1}{c}{$\Delta$ [paired 95\% CI]} \\
    \cline{2-6}\cline{7-7}
    Game & N/E & Logit & sc. & pp. & pa
         & $\Delta$ [95\% CI] \\
    \hline
    SC1  & .1910 & .0183 & .0056 & .0019
         & \textbf{.0010}
         & $-.0173\,[-.0364,-.0004]$ \\
    SC2  & .0947 & .0053 & .0068 & .0051
         & \textbf{.0045}
         & $-.0008\,[-.0140,+.0140]$ \\
    SC3  & .0952 & \textbf{.0002} & .0356 & .0347
         & .0261
         & $+.0258\,[+.0124,+.0385]$ \\
    SC4  & .0344 & \textbf{.0002} & .0147 & .0147
         & .0093
         & $+.0091\,[-.0012,+.0183]$ \\
    SC5  & .0132 & \textbf{.0011} & .0062 & .0044
         & .0032
         & $+.0021\,[-.0074,+.0089]$ \\
    SC6  & .0039 & .0021 & .0017 & .0008
         & \textbf{.0005}
         & $-.0016\,[-.0053,+.0019]$ \\
    SC7  & .4286 & .0196 & .0035 & .0017
         & \textbf{.0008}
         & $-.0188\,[-.0328,-.0077]$ \\
    SC8  & .0604 & .0286 & .0008 & .0008
         & \textbf{.0007}
         & $-.0279\,[-.0516,-.0083]$ \\
    SC9  & .0345 & .0111 & .0028 & .0010
         & \textbf{.0004}
         & $-.0108\,[-.0322,+.0234]$ \\
    SC10 & .0418 & .0226 & .0045 & .0001
         & \textbf{.0000}
         & $-.0226\,[-.0369,-.0019]$ \\
    SC11 & .0159 & .0014 & .0035 & .0032
         & \textbf{.0013}
         & $-.0001\,[-.0090,+.0070]$ \\
    SC12 & .0027 & .0008 & .0008 & .0005
         & \textbf{.0003}
         & $-.0005\,[-.0041,+.0014]$ \\
    \hline
    Mean & .0847 & .0093 & .0072 & .0057
         & \textbf{.0040}
         & $\mathbf{-.0053\,[-.0092,-.0008]}$ \\
    \hline
    \end{tabular}}
  \end{table}

\begin{figure}[t]
    \centering
    \includegraphics[width=.62\textwidth]
    {\detokenize{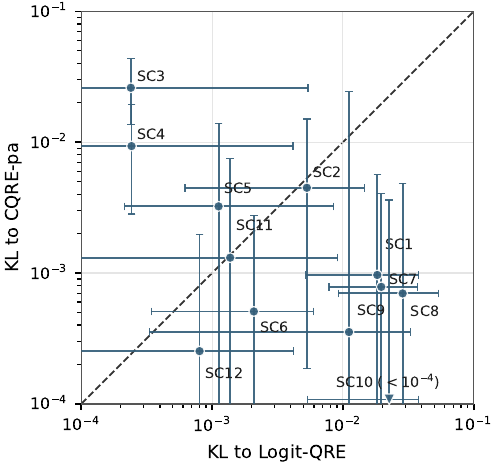}}
    \caption{Fit to the twelve Selten--Chmura human games. Points show
    full-sample KL divergence, and bars show marginal 95\% bootstrap
    intervals obtained from paired refits. Both axes use logarithmic
    scales; points below the diagonal favor CQRE-pa. Intervals extending
    below $10^{-4}$ are clipped at the plot boundary, and the CQRE-pa
    estimate for SC10 is displayed there because it is below $10^{-4}$.}
    \label{fig:selten-human-kl}
  \end{figure}

\subsection{All-12 Selten--Chmura Extension}
\label{app:all12-results}

Table~\ref{tab:app-all12-population} supplements
Table~\ref{tab:selten-populations} by reporting the bootstrap mean as
well as the point estimate.  The two directional HAI rows separate
targets in which the AI occupies the row role from those in which the
human occupies the row role.  Because every HAI target is constructed
as the product of the corresponding human and AI marginals, these rows
do not represent observed human--AI interaction.

\begin{table}[ht]
  \centering
  \small
  \setlength{\tabcolsep}{5pt}
  \caption{All-12 Selten--Chmura inference for the primary
  $\beta\in[0,\infty]$ Logit-QRE comparison.  The point column is the
  equal-target mean of~\eqref{eq:appendix-delta}.}
  \label{tab:app-all12-population}
  \begin{tabular}{@{}lrrrr@{}}
  \hline
  Group & Targets & Point $\Delta$ & Boot. mean & Paired 95\% CI \\
  \hline
  HH             & 12  & $-.00529$ & $-.00500$ & $[-.00915,-.00082]$ \\
  AAI            & 432 & $-.00727$ & $-.00725$ & $[-.00825,-.00627]$ \\
  HAI            & 144 & $+.01115$ & $+.01110$ & $[+.00975,+.01240]$ \\
  \quad AI row   & 72  & $+.01445$ & $+.01437$ & $[+.01227,+.01628]$ \\
  \quad human row& 72  & $+.00785$ & $+.00783$ & $[+.00584,+.00978]$ \\
  \hline
  All            & 588 & $-.00272$ & $-.00271$ & $[-.00354,-.00192]$ \\
  \hline
  \end{tabular}
\end{table}

 Table~\ref{tab:app-all12-by-game} reports the game-level comparisons
  separately for AAI, constructed HAI, and HH targets. Each AAI row
  averages 36 ordered targets, each HAI row averages 12 directional
  targets, and each HH row contains one human target.

  \begin{table}[ht]
    \centering
    \scriptsize
    \setlength{\tabcolsep}{5pt}
    \caption{All-12 extension by game and population under the primary
    $\beta\in[0,\infty]$ Logit-QRE specification. Negative $\Delta$
    favors CQRE-pa. Intervals are synchronized paired 95\% bootstrap
    intervals.}
    \label{tab:app-all12-by-game}

    \textit{(a) AI--AI targets}\\[2pt]
    \begin{tabular}{@{}lrrl@{}}
    \hline
    Game & Point $\Delta$ & Bootstrap mean & Paired 95\% CI \\
    \hline
    SC1  & $-.03812$ & $-.03873$ & $[-.04276,-.03482]$ \\
    SC2  & $+.00366$ & $+.00328$ & $[-.00147,+.00806]$ \\
    SC3  & $+.01903$ & $+.01911$ & $[+.01652,+.02184]$ \\
    SC4  & $+.02232$ & $+.02221$ & $[+.01962,+.02476]$ \\
    SC5  & $-.00221$ & $-.00225$ & $[-.00397,-.00044]$ \\
    SC6  & $+.01360$ & $+.01340$ & $[+.01206,+.01480]$ \\
    SC7  & $-.10924$ & $-.10804$ & $[-.11456,-.10129]$ \\
    SC8  & $-.01371$ & $-.01325$ & $[-.01645,-.00989]$ \\
    SC9  & $+.00891$ & $+.00882$ & $[+.00571,+.01179]$ \\
    SC10 & $+.00218$ & $+.00222$ & $[-.00002,+.00460]$ \\
    SC11 & $+.00951$ & $+.00935$ & $[+.00706,+.01158]$ \\
    SC12 & $-.00321$ & $-.00317$ & $[-.00477,-.00159]$ \\
    \hline
    \end{tabular}

    \vspace{5pt}
    \textit{(b) Constructed human--AI targets}\\[2pt]
    \begin{tabular}{@{}lrrl@{}}
    \hline
    Game & Point $\Delta$ & Bootstrap mean & Paired 95\% CI \\
    \hline
    SC1  & $-.00487$ & $-.00545$ & $[-.01219,-.00057]$ \\
    SC2  & $+.03206$ & $+.03214$ & $[+.02595,+.03838]$ \\
    SC3  & $+.03152$ & $+.03158$ & $[+.02673,+.03660]$ \\
    SC4  & $+.02649$ & $+.02632$ & $[+.02130,+.03081]$ \\
    SC5  & $+.00601$ & $+.00610$ & $[+.00188,+.01148]$ \\
    SC6  & $+.00838$ & $+.00848$ & $[+.00757,+.00943]$ \\
    SC7  & $-.00605$ & $-.00647$ & $[-.01058,-.00290]$ \\
    SC8  & $-.00682$ & $-.00640$ & $[-.01267,-.00138]$ \\
    SC9  & $+.01533$ & $+.01532$ & $[+.00876,+.02282]$ \\
    SC10 & $+.01168$ & $+.01136$ & $[+.00668,+.01563]$ \\
    SC11 & $+.01258$ & $+.01270$ & $[+.01011,+.01529]$ \\
    SC12 & $+.00749$ & $+.00755$ & $[+.00633,+.00893]$ \\
    \hline
    \end{tabular}

    \vspace{5pt}
    \textit{(c) Human--human targets}\\[2pt]
    \begin{tabular}{@{}lrrl@{}}
    \hline
    Game & Point $\Delta$ & Bootstrap mean & Paired 95\% CI \\
    \hline
    SC1  & $-.01735$ & $-.01650$ & $[-.03644,-.00040]$ \\
    SC2  & $-.00082$ & $-.00057$ & $[-.01404,+.01400]$ \\
    SC3  & $+.02583$ & $+.02558$ & $[+.01244,+.03847]$ \\
    SC4  & $+.00910$ & $+.00887$ & $[-.00120,+.01827]$ \\
    SC5  & $+.00211$ & $+.00142$ & $[-.00743,+.00886]$ \\
    SC6  & $-.00158$ & $-.00151$ & $[-.00528,+.00190]$ \\
    SC7  & $-.01885$ & $-.01931$ & $[-.03282,-.00774]$ \\
    SC8  & $-.02792$ & $-.02669$ & $[-.05161,-.00825]$ \\
    SC9  & $-.01077$ & $-.00859$ & $[-.03156,+.02340]$ \\
    SC10 & $-.02260$ & $-.02160$ & $[-.03679,-.00189]$ \\
    SC11 & $-.00006$ & $-.00039$ & $[-.00904,+.00695]$ \\
    SC12 & $-.00055$ & $-.00075$ & $[-.00409,+.00141]$ \\
    \hline
    \end{tabular}
  \end{table}

\subsection{Larger-Action Games}
\label{app:multigame-results}
\begin{table}[t]
    \centering
    \small
    \setlength{\tabcolsep}{5pt}
    \caption{Continuous CQRE-pa versus stored Logit-QRE fits over
    $\beta\in[0,20]$. The table reports mean KL and
    $\Delta=\KL_{\mathrm{CQRE\text{-}pa}}
    -\KL_{\mathrm{Logit\text{-}QRE}}$ with synchronized
    $B=1000$ fixed-prediction bootstrap intervals.
    Negative $\Delta$ favors CQRE-pa. C/U/L denotes CQRE-pa wins,
    unresolved comparisons, and Logit-QRE wins.}
    \label{tab:multigame-results}

    \begin{tabular}{@{}lrrrr@{}}
    \hline
    Game & CQRE-pa & Logit-QRE
         & $\Delta$ [paired 95\% CI] & C/U/L \\
    \hline
    \texttt{sm\_n3} & .6157 & 10.6130
      & $-9.9974\,[-10.3512,-9.6245]$ & 18/0/3 \\
    \texttt{sm\_n5} & .8648 & 1.0050
      & $-.1402\,[-.1438,-.1279]$ & 21/0/0 \\
    \texttt{sm\_n7} & 1.2090 & 1.3872
      & $-.1782\,[-.1837,-.1672]$ & 21/0/0 \\
    \texttt{rank1\_n5} & 1.2082 & 1.3948
      & $-.1867\,[-.1907,-.1784]$ & 21/0/0 \\
    \texttt{rank2\_n5} & 1.3358 & 1.0898
      & $+.2460\,[+.2408,+.2641]$ & 1/4/16 \\
    \texttt{rank3\_n5} & 1.2924 & 7.4076
      & $-6.1153\,[-6.2928,-5.9284]$ & 18/0/3 \\
    \texttt{rank4\_n5} & 1.0759 & 1.1351
      & $-.0592\,[-.0675,-.0476]$ & 11/8/2 \\
    \hline
    All & 1.0859 & 3.4332
      & $-2.3473\,[-2.4064,-2.2859]$ & 111/12/24 \\
    \hline
    \end{tabular}
  \end{table}

The seven-game results in Table~\ref{tab:multigame-results} use the continuous CQRE-pa point fits and the Logit-QRE fits over
$\beta\in[0,20]$.  The associated $B=1000$ intervals are a
\emph{fixed-prediction replay}: each synchronized bootstrap replicate resamples the observations and reevaluates KL against the stored predictions, without refitting either model. 

In Table~\ref{tab:multigame-results}, a CQRE-pa win means that the target-specific paired interval lies below zero, a Logit-QRE win means that it lies above zero, and all other comparisons are unresolved.

 \begin{figure}[t]
    \centering

    \includegraphics[width=.32\textwidth]
    {\detokenize{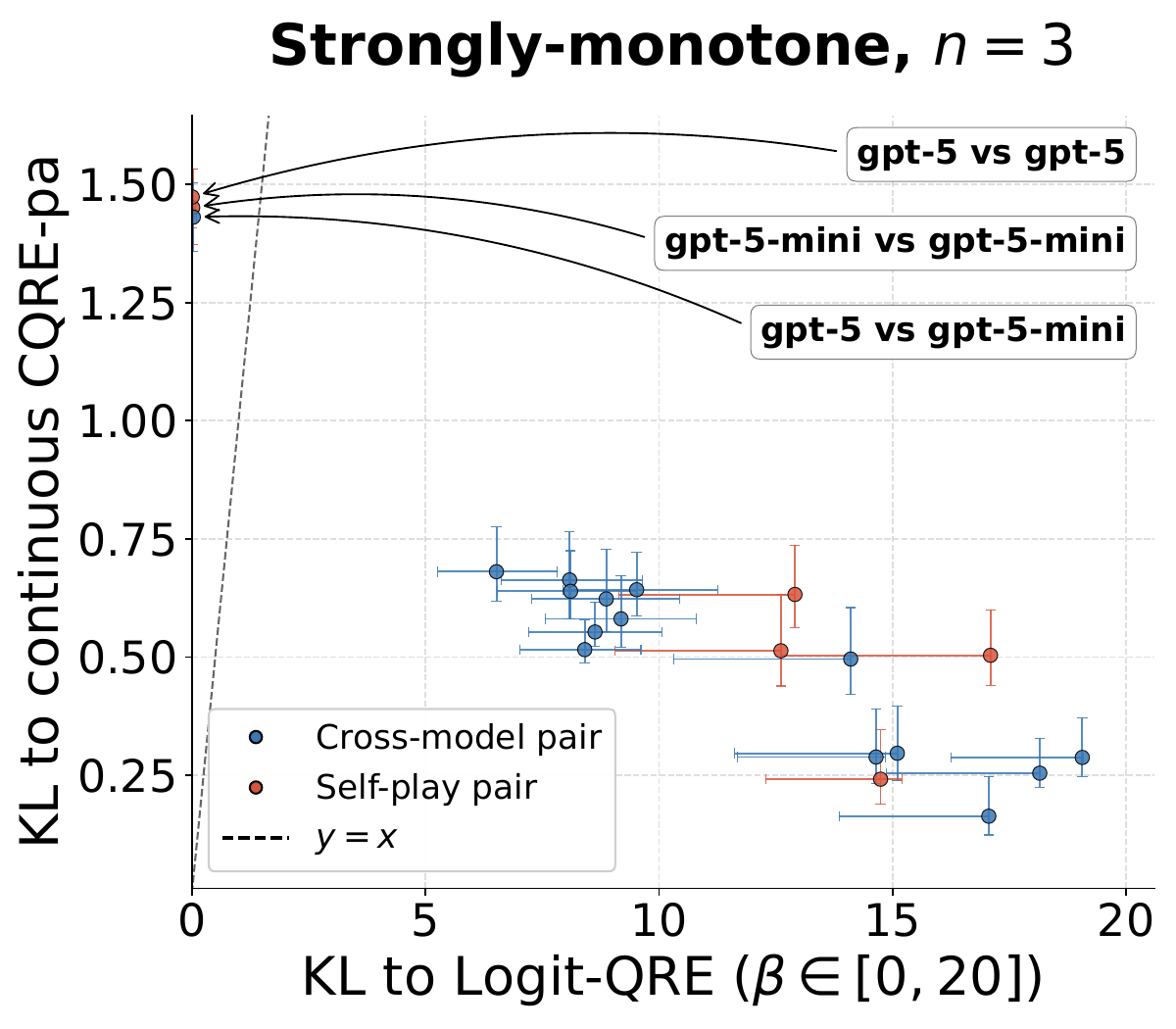}}
    \hfill
    \includegraphics[width=.32\textwidth]
    {\detokenize{figures/aistats_independent/2player/scatter_sm_n5.pdf}}
    \hfill
    \includegraphics[width=.32\textwidth]
    {\detokenize{figures/aistats_independent/2player/scatter_sm_n7.pdf}}

    \par\vspace{-1mm}

    \includegraphics[width=.32\textwidth]
    {\detokenize{figures/aistats_independent/2player/scatter_rank1_n5.pdf}}
    \hfill
    \includegraphics[width=.32\textwidth]
    {\detokenize{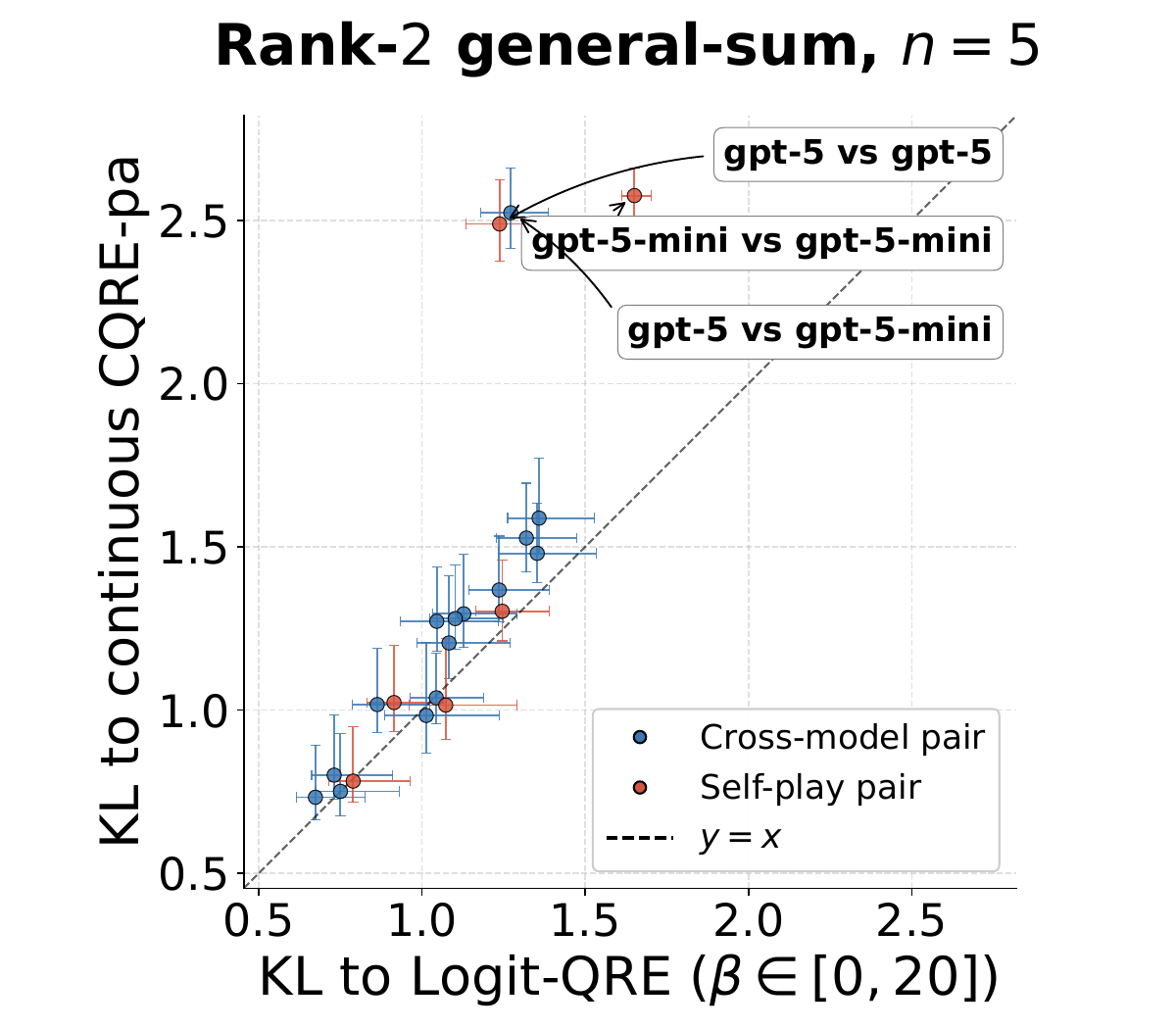}}
    \hfill
    \includegraphics[width=.32\textwidth]
    {\detokenize{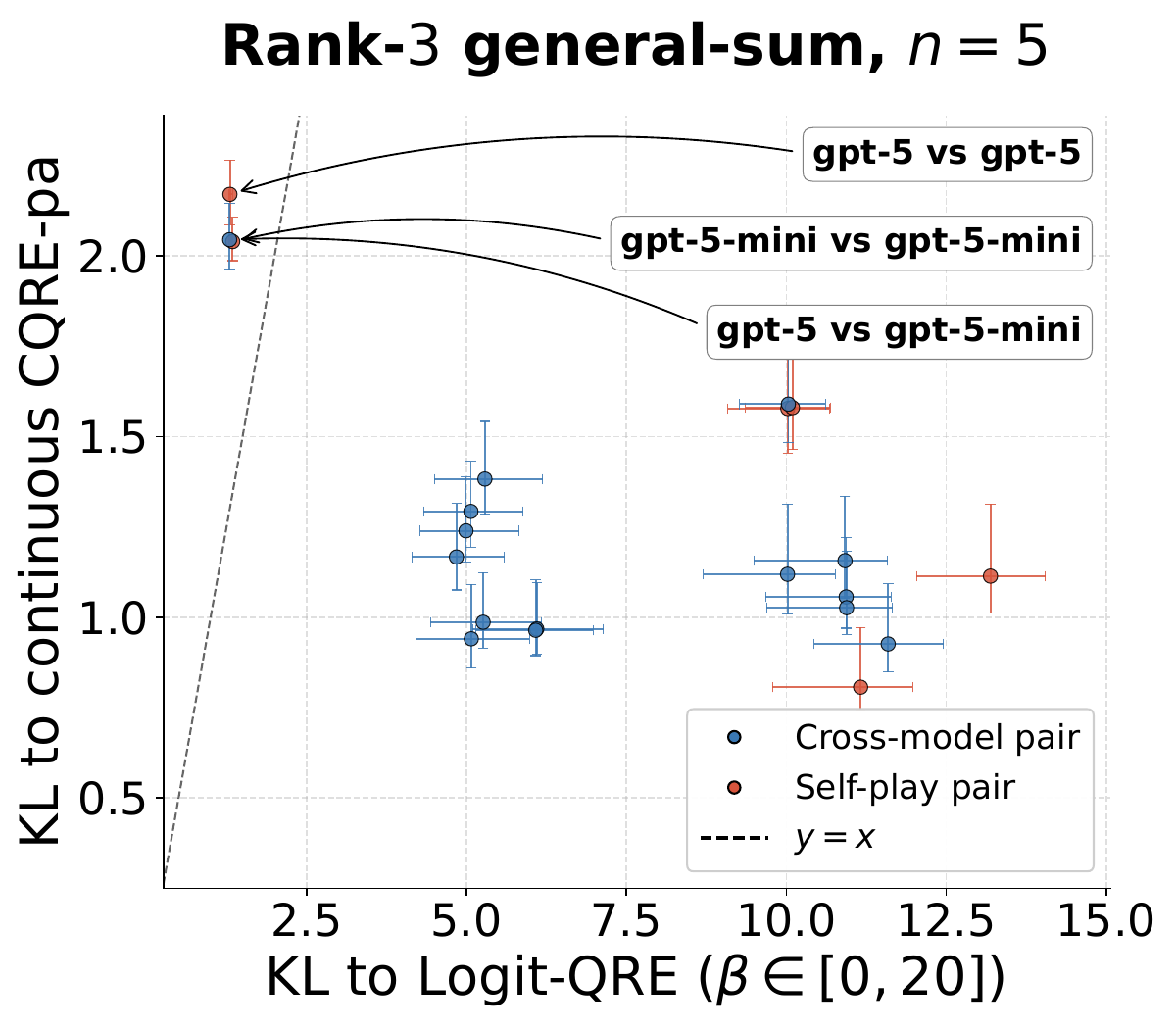}}

    \par\vspace{-1mm}

    \includegraphics[width=.32\textwidth]
    {\detokenize{figures/aistats_independent/2player/scatter_rank4_n5.pdf}}

    \caption{Larger-action comparisons using continuous CQRE-pa and
    stored Logit-QRE fits over $\beta\in[0,20]$. Each panel contains
    21 model pairings; points below the diagonal favor CQRE-pa.
    CQRE-pa has lower mean KL in six of the seven games and sweeps
    \texttt{sm\_n5}, \texttt{sm\_n7}, and \texttt{rank1\_n5};
    \texttt{rank2\_n5} is the principal inversion.}
    \label{fig:multigame-all}
  \end{figure}

  \paragraph{GPT-5-family subgroup.}
  The GPT-5-family subgroup consists of the three self/near-self dyads
  formed by \texttt{gpt-5} and \texttt{gpt-5-mini} in each game, yielding
  21 targets overall. Separating these targets from the other 126 gives
  mean $\Delta=+0.25$ for the GPT-5-family subgroup and $-2.78$ for the
  remaining targets. GPT-5-family play is more concentrated in every
  game, and its mean CQRE-pa KL is 1.7--3.3 times the corresponding
  within-game mean for other targets.

  This concentration is especially pronounced in \texttt{sm\_n3}, where
  the three GPT-5-family targets assign 98--99\% of their mass to the
  unique pure Nash profile $(A,A)$. Logit-QRE can approach this point mass
  as $\beta$ increases, whereas the maximum-entropy CQRE-pa prediction
  retains probability on multiple profiles. Logit-QRE therefore wins all
  three GPT-5-family comparisons in both \texttt{sm\_n3} and
  \texttt{rank2\_n5}. Concentration alone does not determine the result,
  however: CQRE-pa wins all 12 such comparisons across
  \texttt{sm\_n5}, \texttt{sm\_n7}, \texttt{rank1\_n5}, and
  \texttt{rank4\_n5}. The concentration pattern is therefore a useful
  diagnostic, but game structure also materially affects relative fit.

\subsubsection{Extended Logit-QRE Domain}
  \label{app:logit-sensitivity}

  We separately extend the Logit-QRE domain to
  $\beta\in[0,\infty]$, where $\beta=\infty$ denotes the limiting
  mixed-Nash prediction. Within the matched implementation, extending
  the domain leaves all target-level verdicts unchanged; only
  \texttt{rank3\_n5} changes at the displayed game-level precision.
  CQRE-pa remains favored on \texttt{sm\_n5}, with mean KL $.8648$
  versus $.9031$ and $\Delta=-.0383$ with 95\% interval
  $[-.0486,-.0351]$. It also remains favored on
  \texttt{rank4\_n5}, with mean KL $1.0759$ versus $1.1351$ and
  $\Delta=-.0592$ with 95\% interval $[-.0717,-.0520]$.

\subsection{Three-Player Paired-Refit Results}
\label{app:three-player-results}

The three-player analysis uses all 56 model triplets with replacement.
For each of the $B=1000$ synchronized bootstrap
replicates, both CQRE-pa and the Logit-QRE comparator are refit.  Table~\ref{tab:app-three-player} reports the verified aggregate.  Every one of the 56 target-specific
paired intervals lies strictly below zero.

\begin{table}[ht]
  \centering
  \small
  \setlength{\tabcolsep}{7pt}
  \caption{Verified paired-refit results for \texttt{sm\_3p3a}.}
  \label{tab:app-three-player}
  \begin{tabular}{@{}lr@{}}
  \hline
  Quantity & Value \\
  \hline
  Targets & 56 \\
  Bootstrap replicates per target & 1000 \\
  Mean point CQRE-pa KL & .49504 \\
  Mean point Logit-QRE KL & .63165 \\
  Mean point $\Delta$ & $-.13661$ \\
  Bootstrap mean $\Delta$ & $-.13806$ \\
  Aggregate paired 95\% CI & $[-.14322,-.13273]$ \\
  CQRE-pa / unresolved / Logit-QRE & 56 / 0 / 0 \\
  \hline
  \end{tabular}
\end{table}

\subsection{Effective Slack Dimension}
\label{app:slack-rank}
\label{app:slack-dimension}

For each strongly-monotone game size, let $G\in\Real^{21\times 2n}$
contain the componentwise-minimal fitted CQRE-pa slack vectors, one per
model pairing, and center every column of $G$.  If
$s_1\geq s_2\geq\cdots$ are the singular values of the centered matrix,
we define
\begin{equation}
 r_{90}:=\min\left\{r:
 \frac{\sum_{k=1}^{r}s_k^2}{\sum_k s_k^2}\geq .90\right\}.
 \label{eq:app-r90}
\end{equation}
Table~\ref{tab:app-slack-rank} and Figure~\ref{fig:gamma-scree} report the fitted-rank results and singular-value spectra. Across the eight games, 90\% of centered slack variation is captured by 2--7 singular directions, corresponding to 22--50\% of the nominal dimension.

\begin{table}[ht]
  \centering
  \small
  \setlength{\tabcolsep}{8pt}
  \caption{Effective dimension of fitted per-action slack vectors.}
  \label{tab:app-slack-rank}
  \begin{tabular}{@{}lrrrr@{}}
  \hline
  Game & Nominal dim. & $r_{90}$ & $r_{90}/(2n)$ & $s_1^2/\sum_k s_k^2$ \\
  \hline
  \texttt{sm\_n3}  & 6  & 2 & 33.3\% & .634 \\
  \texttt{sm\_n4}  & 8  & 4 & 50.0\% & .612 \\
  \texttt{sm\_n5}  & 10 & 3 & 30.0\% & .541 \\
  \texttt{sm\_n6}  & 12 & 4 & 33.3\% & .690 \\
  \texttt{sm\_n7}  & 14 & 6 & 42.9\% & .493 \\
  \texttt{sm\_n8}  & 16 & 7 & 43.8\% & .381 \\
  \texttt{sm\_n9}  & 18 & 4 & 22.2\% & .556 \\
  \texttt{sm\_n10} & 20 & 7 & 35.0\% & .321 \\
  \hline
  \end{tabular}
\end{table}

\begin{figure}[t]
  \centering
  \includegraphics[width=.46\textwidth]
  {\detokenize{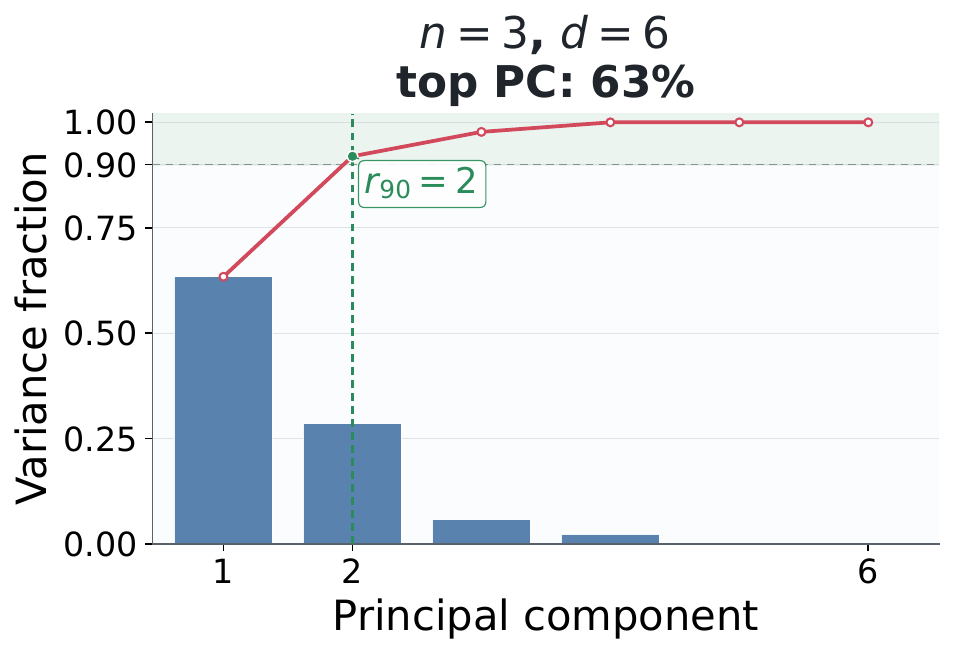}}
  \hfill
  \includegraphics[width=.46\textwidth]
  {\detokenize{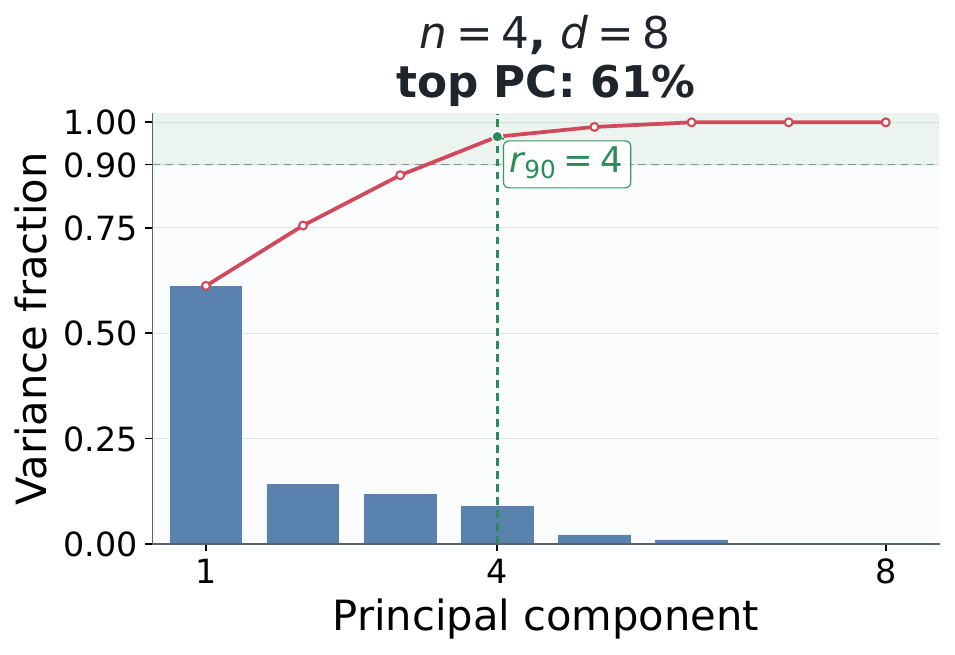}}

  \par\vspace{-1mm}

  \includegraphics[width=.46\textwidth]
  {\detokenize{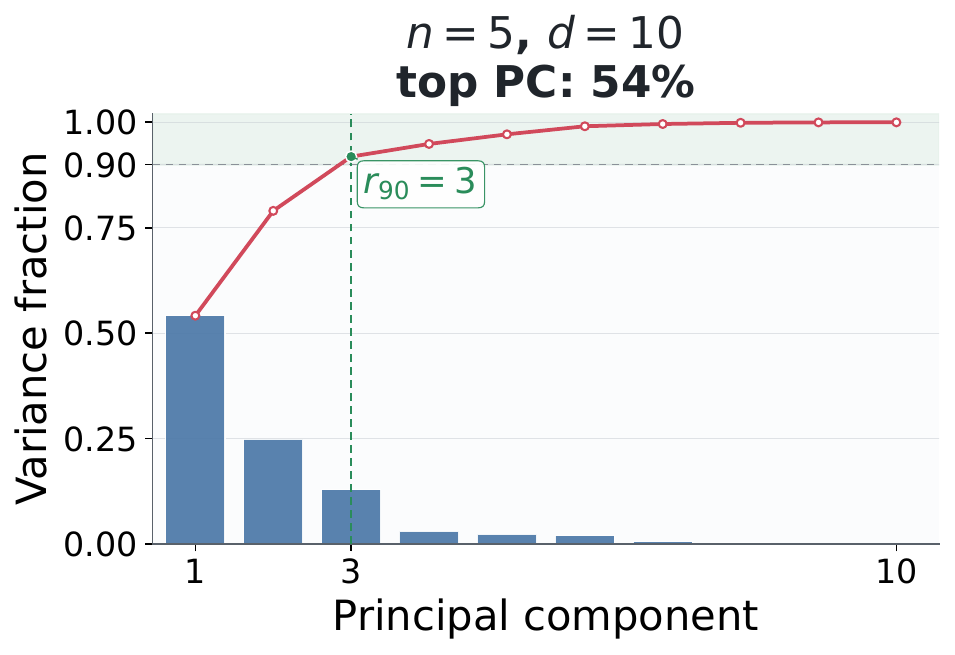}}
  \hfill
  \includegraphics[width=.46\textwidth]
  {\detokenize{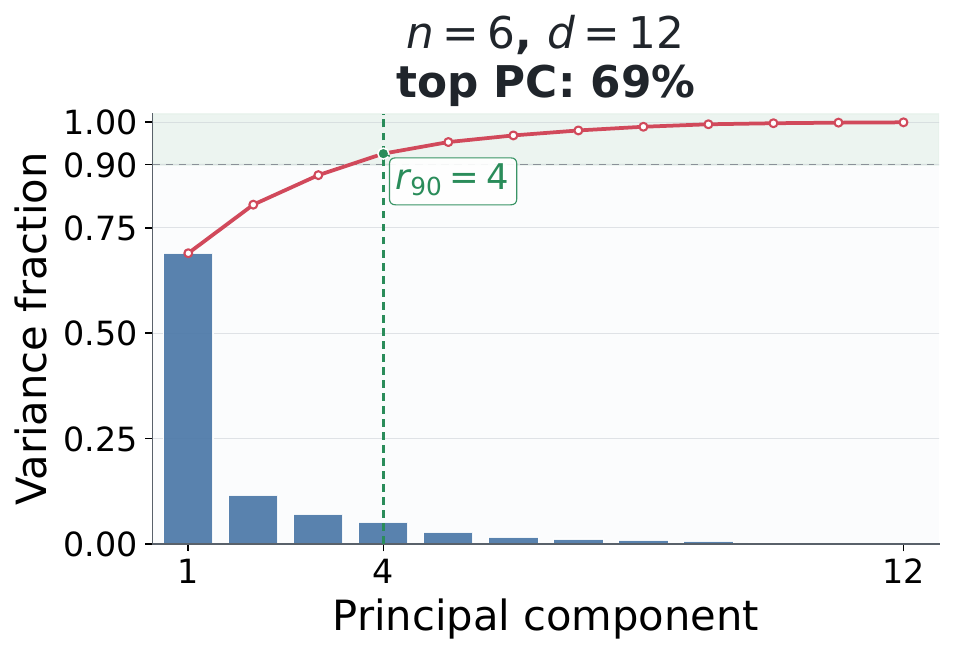}}

  \par\vspace{-1mm}

  \includegraphics[width=.46\textwidth]
  {\detokenize{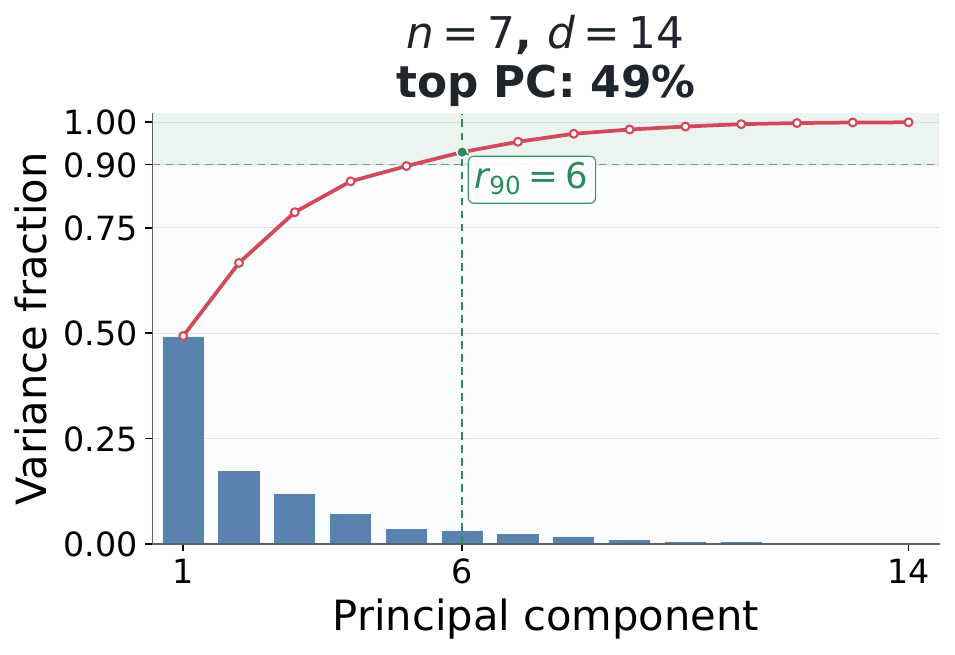}}
  \hfill
  \includegraphics[width=.46\textwidth]
  {\detokenize{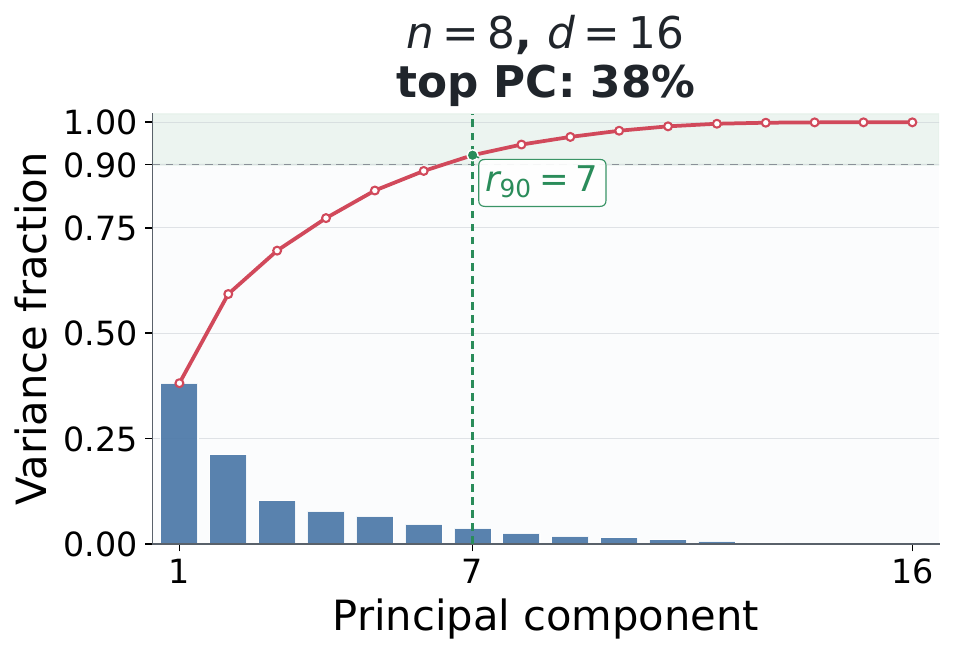}}

  \par\vspace{-1mm}

  \includegraphics[width=.46\textwidth]
  {\detokenize{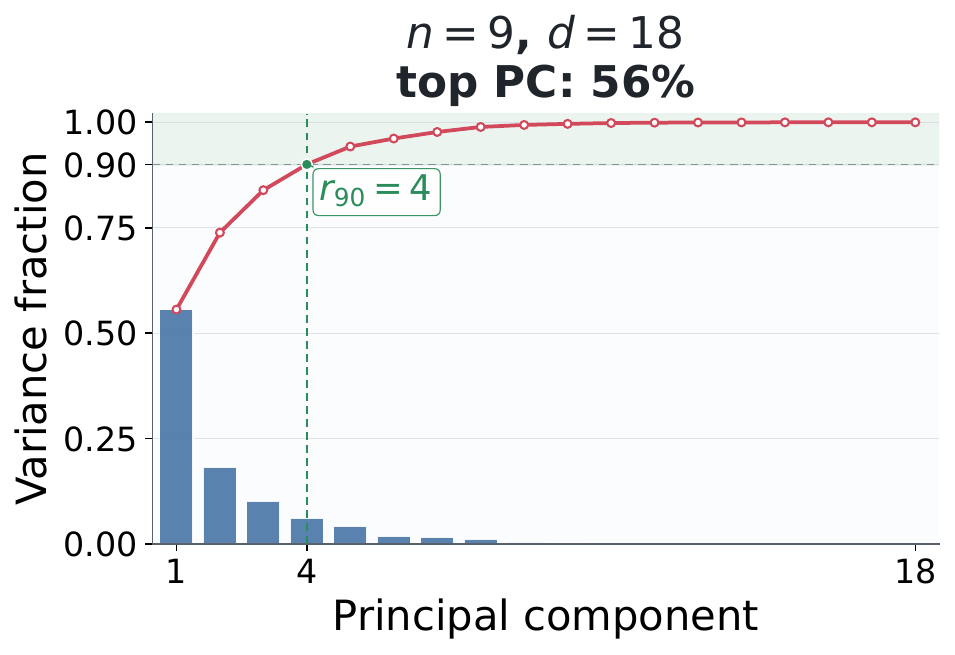}}
  \hfill
  \includegraphics[width=.46\textwidth]
  {\detokenize{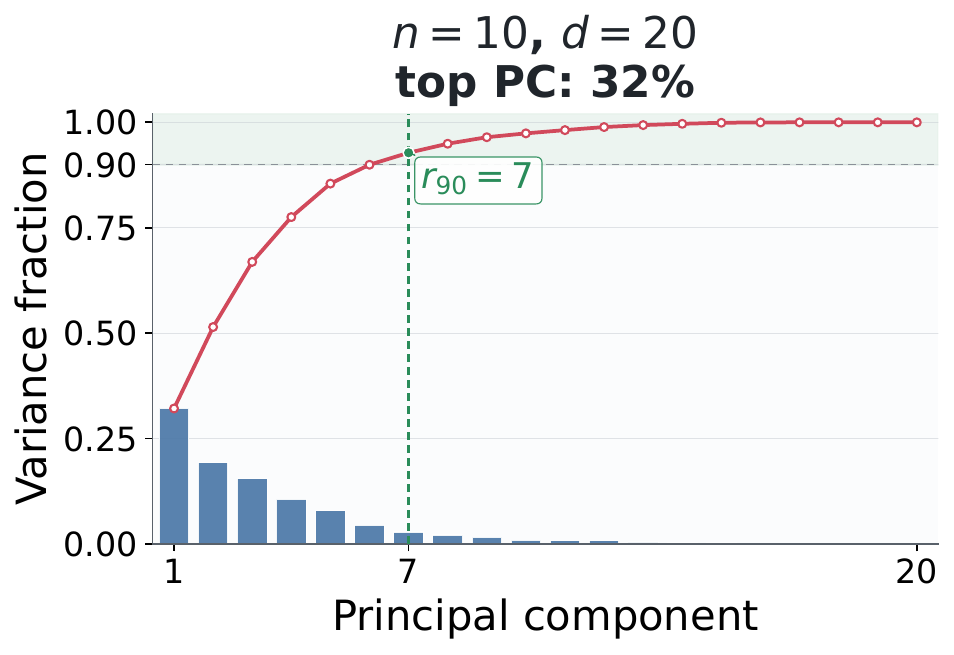}}
  \caption{Singular-value spectra of the centered
  componentwise-minimal CQRE-pa slack vectors. Each panel contains
  21 model-pair targets for one strongly-monotone game. Bars show
  per-component variance fractions, the red curve shows cumulative
  variance, the dashed horizontal line marks 90\%, and the green line
  marks $r_{90}$.}
  \label{fig:gamma-scree}
  \end{figure}

\end{document}